\documentclass[11pt,a4paper]{article}
\usepackage[T1]{fontenc}

\usepackage[a4paper,top=2cm,bottom=2cm,left=2cm,right=2cm,marginparwidth=1.75cm]{geometry}

\usepackage{amsmath}
\usepackage{amssymb}
\usepackage{graphicx}
\usepackage{hyperref}
\usepackage{booktabs}
\usepackage{diagbox}
\usepackage{multirow}
\usepackage{enumitem}
\usepackage{xcolor}
\usepackage{amsthm}
\usepackage{caption}
\usepackage{subcaption}

\usepackage{indentfirst}

\usepackage[affil-it]{authblk}

\usepackage{natbib}

\newcommand{\gd}{\hat{\lambda}}
\newcommand{\lam}[1]{\lambda_{#1, r_l}}
\newcommand{\lambdas}{\lam{1}, \ldots, \lam{k}}
\newcommand{\prob}{\mathbb{P}}
\newtheorem{theorem}{Theorem}
\newcommand{\supp}[1]{\mathrm{supp}(#1)}

\title{\bfseries \normalsize Balancing Profit and Traveller Acceptance in Ride-Pooling Personalised Fares}

\author[1,2]{Micha\l~Bujak*}
\author[1]{Rafa\l~Kucharski}

\affil[1]{Faculty of Mathematics and Computer Science, Jagiellonian University}
\affil[2]{Doctoral School of Exact and Natural Sciences, Jagiellonian University}

\date{\vspace{-5ex}}

\begin{document}
\maketitle

\section*{Abstract}
In a ride-pooling system, travellers experience discomfort associated with a detour and a longer travel time, which is compensated with a sharing discount.
Most studies assume travellers receive either a flat discount or, in rare cases, a proportional to the inconvenience. 
We show the system benefits from individually tailored fares.
We argue that fares that optimise an expected profit of an operator also improve system-wide performance if they include travellers' acceptance. 

Our pricing method is set in a heterogeneous population, where travellers have varying levels of value-of-time and willingness-to-share, unknown to the operator.
A high fare discourages clients from the service, while a low fare reduces the profit margin. 
Notably, a shared ride is only realised if accepted by all co-travellers (decision is driven by the latent behavioural factors). 

Our method reveals intriguing properties of the shareability topology.
Not only identifies rides efficient for the system and supports them with reduced fares (to increase their realisation probability), but also identifies travellers unattractive for the system (e.g. due to incompatibility with other travellers) and effectively shifts them to private rides via high fares.
Unlike in previous methods, such approach naturally balances the travellers satisfaction and the profit maximisation.
With an experiment set in NYC, we show that this leads to significant improvements over the flat discount baseline: the mileage (proxy for environmental externalities) is reduced by $4.5\%$ and the operator generates more profit per mile (over $20\%$ improvement).
We argue that ride pooling systems with fares that maximise profitability are more sustainable and efficient if they include travellers' satisfaction.

\textbf{Keywords}: ride-pooling, personalised pricing, individual discounts

\section{Introduction}
Ride-pooling is a door-to-door service similar to a standard ride-hailing with a caveat that, if the requested trips have similar paths and are requested at the similar time, they are pooled together and travellers share parts of their trips. 
We can distinguish three parties in the ride-pooling system.

First, when we consider travellers' perspective, we notice sharing rides is inherently associated with certain discomfort.
The travelled path is longer (other passengers are picked up and dropped off along the way), the starting time may be shifted and we share the vehicle space with other drivers. 
The sensitivity for a prolonged travel time is expressed by a value of time.
Sharing the space of the vehicle is another subjective inconvenience, usually referred to as (un)willingness to share (in the study, we will refer to it as a penalty for sharing). 
To compensate for the longer travel time and the inconvenience, travellers are offered a monetary incentive. 
Their fare is reduced compared to a private ride by a factor denoted as a sharing discount. 

Second perspective is at the system level (congestion, emission, sustainability etc.). 
When we discuss shared rides, we focus on a reduced mileage due to increased occupancy of pooled rides. 
It yields natural benefits such as reduction of carbon footprint and traffic. 
However, while these effects are highly appreciated, they yield little incentive to introduce the ride-pooling by a profit-driven operator.

The third perspective, often neglected, is the operator. 
For a commercial service provider to launch and maintain the ride-pooling operations, they need to be profitable. 
Moreover, for an operator providing ride-hailing service, the complementary sharing mode should surpass the already established private-ride model.
The operator naturally aims to achieve the highest possible profit. 
Simultaneously, the profit is obtained only if the clients engage in the service, which requires critical mass of both supply and demand. 

We propose an individualised pricing model set in a probabilistic setting.
While we explicitly maximise the operator's perspective, we show that it directly affects the two remaining parts of the service. 
Our measure (the expected profitability) optimises revenue per mileage.
As a result, the vehicle mileage is implicitly reduced, while ramifications for travellers are more complex. 
Clients who are valuable for the service performance are recognised with significantly lower fare and their satisfaction increases while others are discouraged from pooling in favour of the private ride-hailing (higher fare).

\subsection{Background}\label{sec:background}
On-demand services are appreciated thanks to personalised experience and swift commute similar to private car experience (without the potential parking issue). 
While such rides are enjoyed by travellers, they generate substantial traffic and emission. 
One way to reduce the negative effects is via introduction of the ride-pooling. 
\cite{fagnant2014travel} analysis show that a single shared vehicle can substitute eleven private cars.
The benefits of partial adoption of ride-pooling on top of private ride-hailing are discussed by \cite{engelhardt2019quantifying} and \cite{martinez2015agent}. 
They show that not only the system benefits but also travellers. 
The environmental benefits are analysed by \cite{martinez2017assesing}, while the travellers' experience is further discussed by \cite{zhang2021pool}.

Ride-pooling is sensitive to many variables. 
First, we should highlight the demand structure. 
The impact of the spatiotemporal distribution was analysed by \cite{zwick2022ride} for the case study of two German cities. 
\cite{soza2022shareability} analysed how synthetic demand distribution correlates with the system's performance. 
\cite{shulika2024spatiotemporal} analysed correlation between the demand distribution (both spatiotemporal and density) and KPIs for the widely-used case of New York City. 
The complex relations of how the two-sided market reaches critical mass are addressed in \cite{ghasemi2024modelling}.

Ride-pooling service can only exist if it is appreciated by the travellers. 
To determine whether clients are interested, we need to understand their needs. 
The modal preference was studied recently by \cite{gervzinivc2023potential} for urban areas. 
\cite{krueger2016preferences} investigates impact of the value of time (split between in and out of the vehicle) on the choice between three alternatives: taxi, ride-pooling and public transport. 
\cite{lavieri2019modeling} analysed the extend to which people are eager to share rides (the willingness to share). 
The value of time and the value of reliability in on-demand services were surveyed and analysed by \cite{alonso2020value} and also specifically for ride-pooling \cite{alonso2021determinants}. 
\cite{chavis2017development} found main variables determining the mode choice in a general purpose flexible-route transit systems.

To optimally match travellers into shared rides is a complex problem. 
The first widely recognised solution was proposed by \cite{santi2014quantifying}. 
However, the algorithm suffered from the cap to two co-sharing travellers. 
The solution was extended to allow for more travellers in the seminal paper by \cite{alonso2017demand}. 
\cite{shah2020neural} builds on the algorithm and optimises matching with neural networks. 
\cite{bilali2020analytical} introduces shareability shadows to analyse shareability in their online approach. 
\cite{ke2021probability} proposes an offline matching strategy to measure relationships between demand and the service performance.  
\cite{kucharski2020exmas} offers a strategic utility-based algorithm, where travellers engage in a shared ride only if they perceive it as more attractive than alternatives.

In our study, we focus on pricing, a determinant of the economic performance of a ride-pooling system. 
\cite{pandit2019pricing} compares two flat-fare pricing strategies in ride-pooling platforms: fixed and dynamic, i.e. build on current supply-demand ratio.
\cite{li2022pricing} proposes a flat pricing strategy, where the discount is a subject to a successful pairing.
\cite{ke2020pricing} builds an equilibrating pricing model for upfront pricing based on a pairing rate and the supply-demand relation. 
\cite{zhang2021pool} finds supply-demand equilibrium where travellers choose between solo ride, pooling ride and transit.
While these studies conduct a profound analysis on the market level, they offer flat fares for the homogeneous population and therefore maximise the operator's profit in a simplified setting. 
\cite{karaenke2023benefits} proposes an ex-post pricing method that offers an individualised discount proportional to the traveller's discomfort.
The method, however, does not differentiate shared rides with respect to their contribution to the system and only considers the client's side.
\cite{jacob2021ride} analyses ride-pooling pricing in a queuing two-sided model and while authors acknowledge heterogeneity, they limit analysis to two types of travellers.

\subsection{Motivation \& approach}
We argue that the best promotional to popularise the ride-pooling service is by proving its commercial potential.
If the service is recognised as an attractive add-on to the private ride-hailing, it should encourage operators to introduce it on the mass-scale. 
We build a pricing model where all parties of the ride-pooling operations are satisfied.
Our objective is to maximise the per-kilometre revenue (in a probabilistic setting) by individually fitting proper discounts for shared rides.

Hitherto, for pricing shared rides researches used mostly a flat discount or a proportional to the experienced discomfort (\cite{bujak2024ride}). 
Approaches rely on the behavioural homogeneity of the population (disproved by preference studies, e.g. \cite{krueger2016preferences} \cite{lavieri2019modeling}) and neglect the profit-maximisation perspective of the operator. 
In rare cases, the heterogeneity is admitted, yet strongly simplified.
In our study, we address both limitations.
To forge a cohesive framework optimising the ride-pooling system, we recognise travellers heterogeneous by introducing a probabilistic formulation. 
We introduce an optimisation decomposition that makes the process feasible for real size problems.
We reconcile the two opposing goals of profitability maximisation (operator) and attractiveness (travellers) by finding the golden mean.
The probabilistic formulation allows to determine a discount level when clients are likely to accept a shared ride and the operator maintains a high profit. 
To quantify the operator's perspective, we introduce an expected profitability measure as a ratio of the expected revenue to the expected vehicle mileage.
We also propose a generalised formula which is suitable for a variety of market evaluations when system parameters are known. 
To calculate the expected profitability, we account for different scenarios: either travellers find a certain shared rides attractive and accept the offer or they reject and shift to the private alternative.
As a result, we directly improve the operator's economic performance while indirectly enhancing clients' and environmental positions.

\section{Methodology}\label{sec:methodology}
In our pricing framework, depicted in Figure \ref{fig:general_method}, we can distinguish three main stages. 
First, preliminary operations: we collect trip requests, learn the population distribution of behavioural traits and construct a shareability graph. 
Seeking an optimal discount level for all ride-traveller combinations is computationally unfeasible in real-size problems.
Hence, we use an existing algorithm to construct the shareability graph and achieve an efficient decomposition. 
At a second stage, we develop a pricing mechanism at a ride level. 
For each feasible ride (the shareability graph), we scrutinise the discomfort perceived by travellers. 
The discomfort level coupled with a fare determines whether a client accepts the shared ride via acceptance probability function.
We obtain how likely, when offered a certain fare, a traveller is to accept the shared ride. 
A shared ride inherently comprises several travellers who independently perceive the discomfort.
The operator needs to account for ride rejection (acceptance) decisions committed by clients. 
Considering likelihood of certain decisions, the operator seeks such discounts that maximise the objective function. 
To address goals of our study, we propose an expected profitability which combines the likelihood and the revenue from each case (travellers' decisions) and maximises the expected revenue per expected distance.
As an output of the second step, for each ride we obtain a vector of individualised fares and corresponding evaluation (of each ride in the shareability graph). 
Now, we move to the third step. 
We apply integer linear programming techniques to select an offer that maximises the average expected profitability: each traveller receives a single option of a shared ride or they are immediately shifted to a private ride.


\begin{figure}[!ht]
    \centering
    \includegraphics[width=\textwidth]{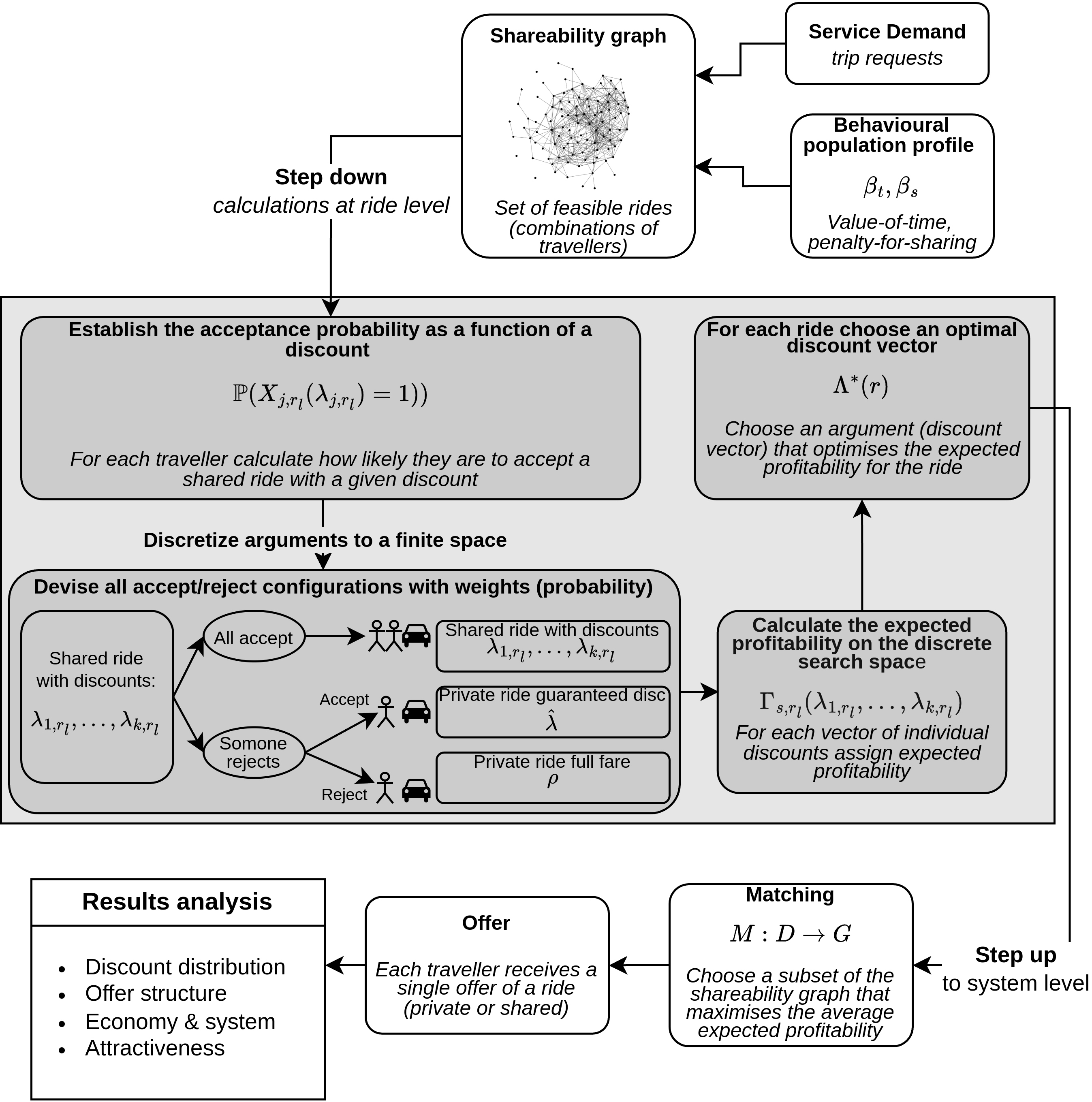}
    \caption{
    Method at glance: We use a batch of trip requests (origin, destination and time) from the population with a known distribution of behavioural parameters. We construct a dense shareability graph to identify all potential pooled rides. Next, we conduct calculations for each ride. We construct the acceptance probability as a function of a fare for each traveller. To calculate the expected profitability (at the ride level), we recognise all acceptance/reject configurations with corresponding weights (probability). We optimise the fare considering the trip characteristics and acceptance probability functions to maximise the expected profitability (objective). Finally, we select an offer, i.e. subset of pooled rides that maximises the objective.
    }
    \label{fig:general_method}
\end{figure}

\subsection{Requests collection and shareability graph construction}\label{sec:requests_shareability}
At the first stage of our algorithm, we establish an input.
We collect trip request and learn population distribution of behavioural traits from external studies.
We calculate quantiles of the behavioural distribution to construct a dense shareability graph.

To allow for a strategic evaluation, we assume the requests are collected prior to the operations (booking).
In such a setting, we have a complete and exact knowledge of the demand. 
Our first task is to create an initial (dense) shareability graph.
The shareability graph is a set of all feasible combinations of travellers, along with the order of pick-ups and drop-offs (formalised in \cite{bujak2023network}). 
To construct the shareability graph (of any degree $\geq 2$), we apply the utility-based algorithm ExMAS proposed by \cite{kucharski2020exmas}, which conducts a hierarchical search shrinking the search space significantly.

The ExMAS algorithm recognises a shared ride as feasible if it is more attractive for a client than the private counterpart. 
The solution is provided with the assumption that there is a flat discount $\lambda$ offered to all sharing a ride and no discount otherwise.
To measure attractiveness (via measuring utility/unattractiveness), we use the following utility formulas (for non-shared and shared rides, respectively):
\begin{subequations}\label{eq:utilities_exmas}
\begin{align}
& U^{ns}_i = -\rho d_i - \beta_t t_i \label{eq:exmas_ns}\\ 
& U^s_{i, r_l} = -(1 - \lambda) \rho d_i - \beta_t \beta_s (\hat{t}_i + \hat{t}^p_i), \label{eq:exmas_s}
\end{align}
\end{subequations}
where $\rho$ stands for fare ($\$ / \mathrm{km} $);
$\beta^t$ and $\beta_s$ denote the behavioural traits: the value of time and the penalty for sharing;
$t_i$ and $\hat{t}_i$ stand for travel time of non-shared and shared ride, respectively, $\hat{t}^p_i$ for a pick-up delay due to pooling (for a given candidate) and
$d_i$ represents length of the requested trip.

The density of the shareability graph created by the ExMAS algorithm is driven predominantly by three factors: the sharing discount $\lambda$ (flat), the value of time $\beta_t$ and penalty for sharing $\beta_s$ (homogeneous). 
In the pricing process, and later matching, we consider only rides present in the shareability graph.
To ensure that, at this initial stage, we do not dismiss a profitable ride, we assume a high sharing discount $\lambda$ and ride-pooling favourable behavioural parameters ($\beta_t$, $\beta_s$).
To set those parameters, we introduce a quantile level $\alpha$.
We assume we know the distribution of the behavioural characteristics in the population ($\beta_t$ for the value of time and $\beta_s$ for the penalty for sharing).
We calculate parameters $\beta_t^0$ and $\beta_s^0$ such that:
\begin{equation}\label{eq:quantile_flat}
    \prob(\beta_t > \beta_t^0) = \prob(\beta_s > \beta_s^0) = 1 - \alpha.
\end{equation}
$\beta_t^0$ and $\beta_s^0$ represent the lowest $\alpha$ quantile of the population.
There are three main implications of the $\alpha$ level. 
First, the higher the $\alpha$, the more likely we are to omit a ride which can be later perceived as attractive by pooling enthusiasts.
Second, a low $\alpha$ increases computational complexity. 
Third, a high $\alpha$ assures that travellers find an offer attractive and are less likely to reject a shared ride.

\subsection{Acceptance probability}\label{sec:acceptance_probability}
In the previous step, we constructed a shareability graph. 
From a traveller's perspective, the discomfort depends on the travel and waiting times, i.e. is specific to a ride. 
Hence, to quantify the acceptance probability, we step down to a ride level. 
In this section, we construct functions that tell us how likely a traveller is to accept a certain ride when offered a certain fare for each ride in the shareability graph.

The discomfort experienced by an individual can be divided into two categories: objective and subjective. 
For the objective measures, we distinguish fare $\rho$, travel time $\hat{t}_i$ and waiting time $\hat{t}^p_i$.
On the subjective side, we think of behavioural traits (value-of-time $\beta_t$, penalty-for-sharing $\beta_s$).
While the subjective variables do not change the trip's characteristic, they drive the perception of the experience (eq. \ref{eq:exmas_s}). 
The objective measures of a ride can be explicitly calculated and are known to the operator. 
Conversely, the individual characteristics are typically latent.
Hence, to calculate the acceptance probability, we leverage the population distribution of the behavioural traits, which we assume is known.
For example, the operator conducted a preference study or relies on historic data. 
Since the individual characteristics remain latent but the population distribution is known, we introduce a probabilistic formulation.

To formalise, we assume the population distribution of the value of time $\beta_t$ and the penalty for sharing $\beta_s$ is known (e.g. thanks to preference studies).
We consider travellers indistinguishable, i.e. their respective vectors of $\beta_t$ and $\beta_s$ are independent and identically distributed random variables (i.i.d.). 
Utility of a private ride for traveller $i$ is calculated as follows.
\begin{equation}
    U^{ns}_i = -\rho d_i - \beta_t t_i.\label{eq:utility_ns}
\end{equation}
While the formulation looks the same as in construction of the shareability graph (eq. \ref{eq:exmas_ns}), the interpretation is different, because now $\beta_t$ and consequently $U^{ns}_i$ are random variables.
In the probabilistic reformulation we acknowledge travellers' heterogeneity in perceived value-of-time $\beta_t$.

To analyse utility of a shared ride, we need to account for additional factors. 
First, the ride is shared, so we have a second random variable in the equation (penalty for sharing). 
We also expand the original formula (eq. \ref{eq:exmas_s}) to account for the number of co-travellers (different levels of $\beta_s$) as \cite{alonso2021determinants}. 
Second, the utility is a function of a shared discount (our control variable).
The utility of a shared ride $r_l$ for traveller $i$ is calculated as:
\begin{equation}
    U^s_{i, r_l}(\lam{i}) = -(1 - \lam{i}) \rho d_i - \beta_t \beta_{s, k} (\hat{t}_i + \hat{t}^p_i),\label{eq:utility_s}
\end{equation}
where $\lambda_{i, r_l}$ is a personalised discount for traveller $i$ for ride $r_l$ and $\beta_{s, k}$ is a penalty for sharing with $k-1$ co-travellers.
The level of $\lambda_{i, r_l}$ is set by our method in next steps.

In our setting, a traveller is offered a shared ride which he either accepts with the proposed discount or rejects in favour of the private alternative. 
Variable $X_{i, r_l} \in \{0, 1 \}$ indicates whether traveller $i$ chooses the shared ride $r_l$ ($X_{i, r_l}(\lam{i})=1$) if offered $\lam{i}$ discount over the non-shared ride:
\begin{equation}\label{eq:single_accept}
    \prob(X_{i, r_l} (\lam{i}) = 1) = \prob(U^s_{i, r_l}(\lam{i}) - U^{ns}_i > 0).
\end{equation}

If any of co-travellers reject the shared ride, all travellers are served with private rides. 
We assume they make simultaneous and independent decisions.
Hence, the probability that the shared ride $r_l$ comprising $k$ travellers ($1, \ldots, k$) is accepted with vector of personalised discounts $\lambda := (\lambdas)$ is
\begin{equation}\label{eq:prob_accept}
    \prob(X_{r_l}(\lambdas) = 1) = \prob(\Pi_{1 \leq i \leq k} X_{i, r_l}(\lam{i}) = 1) = \Pi_{1 \leq i \leq k} \prob(X_{i, r_l}(\lam{i}) = 1).
\end{equation}

\subsection{Discretization of the discount space}\label{sec:discretization}
At this stage, we should introduce a technique which makes our approach computationally feasible. 
While our approach has an analytical formula, the solution to the problem requires numerical approximation. 
In this section, we describe the problem and our proposed solution. 
In a nutshell, we restrict the search space (potential individual discounts) to a finite set.

We assume behavioural characteristics $\beta_t$ and $\beta_s$ to be continuous.
In our case, the vector $(\beta_t, \beta_s)$ follows multimodal, two-dimensional normal distribution.
Let $\Delta U_{i, r_l}(\lam{i}) = U^s_{i, r_l}(\lam{i}) - U^{ns}_i$. 
We calculate $\Delta U_{i, r_l}(\lam{i})$ according to formulas \ref{eq:utility_ns}, \ref{eq:utility_s} and obtain
\begin{equation}
    \prob(\Delta U(\lam{i})\geq0) = \prob(\beta_t (t_i - \beta_s(\hat{t}_i + \hat{t}^p_i)) \leq \lam{i} \rho d).
\end{equation}
We have a multiplication of multimodal, correlated random variables with a non-zero mean (in $Y := \beta_t (t_i - \beta_s(\hat{t}_i + \hat{t}^p_i)$, the $\beta_t \beta_s$ multiplication).
There is no closed formula to obtain the cumulative distribution function (cdf) of $Y$.
Therefore, we propose a numerical approximation as follows.
We discretize both $\beta_t$ and $\beta_s$. 
They now have known finite supports ($\supp{\beta_t} = \{ t_1, \ldots, t_{n_t} \}$ and $\supp{\beta_s} = \{ s_1, \ldots, s_{n_s} \} $, respectively).
Without loss of generality, assume $t_1 < \ldots < t_{n_t}$ and $s_1 < \ldots < s_{n_s}$.
Cdf of $\beta_t$ is constant on intervals $[t_i, t_{i+1})$ (analogically for $\beta_s$).
Therefore, its sufficient to calculate cdfs' values at the product of $\supp{\beta_t}$ and $\supp{\beta_s}$.
From the values, we calculate $\lam{i}$ levels.
Numerate them as $\lambda_1 < \ldots < \lambda_{n_t \times n_s}$.
Note that
\begin{equation}
    \forall_{\lambda \in [\lambda_i, \lambda_{i+1})} \prob( \Delta U(\lambda)\geq0) = \prob(\Delta U(\lambda_i)\geq0).
\end{equation}
On the intervals, the probability of accepting the ride is constant however, the revenue is the highest at the beginning of the intervals (confer eq. \ref{eq:utility_ns}, \ref{eq:utility_s}). 
Hence, it is sufficient to consider $\lambda_1, \ldots, \lambda_{n_t \times n_s}$ in the optimisation process.
We have control over precision of the approximation by manipulating $n_t$ and $n_s$ (the higher the values, the more precise results).

\subsection{Decision-dependent service realisation}\label{sec:pricing_scheme}
Hitherto, we constructed the shareability graph, defined the acceptance probability functions and discretized the discounts space. 
To properly build our probabilistic measure (expected profitability), we need to account for travellers' decisions. 

In the proposed framework, each client in the system receives a guaranteed discount $\gd$.
The guaranteed discount is an incentive towards shifting from a private ride-hailing to the ride-pooling system. 
Travellers, if compliant with the offered ride, have a reduced fare even if they are not matched with another co-traveller.
If the operator recognises a client as compatible with others, he is offered a personalised sharing discount ($\lam{i} \geq \gd$) to join a shared ride $r_l$.
Otherwise (if not compatible), he travels privately at the reduced fare by the $\gd$. 
Clients, who are offered a shared ride, simultaneously and independently decide to accept or reject.
Notably, we assume that if any of co-travellers opts against the shared ride, the proposed ride is not realised and all travel privately. 
Client who rejects is stripped of the discounted fare.
Ones who accept, travel privately at a reduced fare (maintain the guaranteed discount) as in Figure \ref{fig:general_method}.

\subsection{Objective function and fares optimisation}\label{sec:objective_measure}
Our study resolves around the pricing mechanism. 
In previous sections, we described tools required to optimise the sharing discount. 
In this section, we introduce the expected profitability (calculated at a ride level).
At a ride level, we optimise individual fares to maximise the expected profitability.
Once we completed the optimisation process, each feasible ride has assigned a vector of personalised discount alongside resulting expected profitability (required for matching).

Following the decision-dependent scheme presented in Section \ref{sec:pricing_scheme}, we need to account for accept/reject decisions made by travellers (depicted in Figure \ref{fig:general_method} in the darkened box, lower left). 
Each scenario (combination of accept/reject decisions) is weighted by its probability (eq. \ref{eq:prob_accept}). 
High personalised discounts increase the probability that travellers will opt for the pooled options. 
However, it simultaneously reduces the revenue for the operator.
Hence, we search for a golden mean, where shared rides are attractive for travellers and deliver maximum profit to the operator.
We should note here that personalised sharing discount $\lam{i}$ cannot be lower than the guaranteed discount.

We assume travellers always accept private ride option ($X_i^{ns} = 1$). 
If a client $i$ is not compatible with others, he is offered a non-shared ride with the guaranteed discount. 
Hence, the revenue is
\begin{equation}\label{eq:nd_revenue}
    R_i(\gd) = \rho (1 - \gd) d_i,
\end{equation}
where $\rho$ denotes fare, $\gd$ guaranteed discount and $d_i$ trip distance.
Under our assumption that $X_i^{ns} = 1$ (private ride is always realised), the expected revenue $\gamma_{ns, i}$ is equal to the revenue associated with the single case (traveller always accepts a private ride), i.e.
\begin{equation}
    \gamma_{ns, i}(\gd) = R_i(\gd) * \prob(X_i^{ns} = 1) = R_{ns, i}(\gd).
\end{equation}

Expected revenue computation becomes complex for a shared ride offer, as we must account for different scenarios (Fig. \ref{fig:general_method}). 
We have a single case where everyone accepts the offer and multiple configurations combining travellers who reject (no discount) and those who accept (guaranteed discount).
The expected revenue $\gamma_{s, r_l}$ is a function of a vector of personalised discounts $\lambdas$.
We calculate it on a ride $r_l$ (comprising $k$ travellers) level as follows.

\resizebox{0.97\textwidth}{!}{%
\begin{minipage}{1.2\textwidth}
\label{eq:expected_revenue}
\begin{subequations}
    \begin{align}
    & \gamma_{s, r_l}(\lambdas) = \nonumber \\
    & \prob(\Pi_{1 \leq i \leq k} X_{i, r_l}(\lam{i}) = 1) \left(\sum_{1 \leq i \leq k} R_i(\lam{i}) \right) + \label{eq:e_revenue_s} \\
    & \sum_{0 \leq s < k} \sum_{\substack{i_1, \ldots, i_k = s, \\ i_j \in \{0, 1\}}} \prob(X_{1, r_l}(\lam{1}) = i_1, \ldots, X_{k, r_l}(\lam{i}) = i_k) \left(\sum_{1 \leq j \leq k} i_j * R_j (\gd) + |1 - i_j| * R_j(0)\right) \label{eq:e_revenue_ns},
\end{align}
\end{subequations}
\end{minipage}
}

where $R_i(\lam{i})$ is a revenue from client $i$ effectively participating in the shared ride $r_l$ with personalised discount of $\lam{i}$ (eq. \ref{eq:nd_revenue}).
The equation \ref{eq:expected_revenue} was separated into two main parts corresponding to two main scenarios. 
\ref{eq:e_revenue_s} expresses the profit when all travellers accept the shared ride. 
\ref{eq:e_revenue_ns} is a sum over different combinations of travellers' choices where at least one rejected.
In the second part, we sum over multiple (sub-)scenarios. 
It starts with scenario probability followed by the revenue from travellers who accepted the shared ride ($R_j(\gd)$) and those who rejected ($R_j(0)$).

One can observe that the highest revenue is obtained when $\prob(\Pi_{1 \leq i \leq k} X_{i, r_l}(\lam{i}) = 0)$, because $R_i$ is a decreasing function ($R_i(\lam{i}) \leq R(\hat{\lambda})$).
$\prob(\Pi_{1 \leq i \leq k} X_{i, r_l}(\lam{i}) = 0)$ represents situation where the proposed discounts are very low so that each shared ride is rejected by all travellers.

Despite the expected revenue being the highest for private rides, it does not account for costs.
To assess the true perspective of the operator, one must consider costs associated both with fleet size and fleet operations. 
To offer a fair parameter-less perspective, we compute the profitability\footnote{In Section \ref{sec:generalised_formula} we provide a generalisation.}.
The profitability is a measure of income per vehicle distance.
To account for the probabilistic setting, we introduce the probabilistic counterpart, the expected profitability.
We formulate the expected profitability as for non-shared and shared rides as follows.
\begin{subequations}
\label{eq:profitability_both}
\begin{align}
    \Gamma_{ns, i}(\gd) = \frac{R_{ns, i}(\gd)}{d_i} = \rho (1 - \gd), \label{eq:profitability_ns} \\
    \Gamma_{s, r_l}(\lambdas) = \frac{\gamma_{s, r_l}(\lambdas)}{\psi_{s, r_l}(\lambdas)},
\end{align}
\end{subequations}
where $\psi(\lambdas)$ is the expected distance.
We can calculate the expected distance $\psi(\lambdas)$ similarly to the expected revenue in eq. \ref{eq:expected_revenue}.
Let $d_{s, r_l}$ denote distance travelled by the shared ride $r_l$, then
\begin{equation}\label{eq:e_dist_s}
    \psi(\lambdas) = \prob(\Pi_{1 \leq i \leq k} X_i = 1) d_{s, r_l} + (1 - \prob(\Pi_{1 \leq i \leq k} X_i = 1)) (\sum_{1 \leq i \leq k} d_i).
\end{equation}

The first term is a probability that the shared ride is realised multiplied by its distance.
The second term expresses sum over individual distances when the ride is not realised.

At a ride level, thanks to the discretization process (Sec. \ref{sec:discretization}), we can conduct the exhaustive search. 
Each ride coupled with the optimal fares vector, is mapped to its evaluation (the expected profitability). 
Formally, pricing is a function that assigns each ride a vector of individual discounts.
Let $T(r_l)$ denote travellers who comprise ride $r_l$.
Pricing $\Lambda$ is defined on the shareability graph $G$ as
\begin{equation}\label{eq:pricing_def}
    \Lambda: G \ni r_l \xrightarrow{} (\lam{1}, \ldots, \lam{|T(r_l)|}) \in \mathbb{R}_{\geq 0}^{|T(r_l)|}.
\end{equation}
To find optimal fares, for each ride $r_l \in G$, we find $\lambdas$ such that they maximise the expected profitability, i.e.
\begin{equation}\label{eq:pricing_optim}
    \lambdas = \underset{(\tau_1, \ldots, \tau_{|T(r_l)|}) \in \mathbb{R}_{\geq 0}^{|T(r_l)|}}{\mathrm{argmax}} \Gamma_{s, r_l}(\tau_1, \ldots, \tau_{|T(r_l)|}).
\end{equation}

\subsection{Offer}\label{sec:offer}
In previous steps, we evaluated each ride in the shareability graph and set discounts maximising the expected profitability.
Now, we move a level up. 
Our task is to find an optimal subset of the shareability graph such that each traveller receives a single proposition (offer) for their ride and our objective is maximised. 
We refer to the process as matching and to the solution as offer. 

Each ride has its evaluation, i.e. the expected profitability resulting from the optimised discounts. 
In our setting, the objective function $O$ is the average expected profitability multiplied by ride degree, i.e. $O(r_l, (\lambdas)) = \Gamma_{s, r_l}(\lambdas) * |T(r)|$ (for non-shared ride $\Gamma_{ns, i}$). 
The multiplier assures that we maximise the average expected profitability (otherwise two single rides with score of $1$ are better than the shared alternative with score $1.9$).
We note that objective function $O$ reaches its maximum for the same argument as $\Gamma$.
We can solve the matching problem with the integer linear programming approach (detailed formulation is provided in Section \ref{sec:loc_glob_optimum}).
For the technical implementation, we use the Python library PuLP \citet{mitchell2011pulp}.
This step concludes our analysis. 
The operator constructs an offer and all travellers receive a single shared or non-shared proposition.
Our formulation of the offer maximises the average expected profitability of the system.

\subsection{Optimisation process: local vs global approach}\label{sec:loc_glob_optimum}
Our optimisation process started at a ride level and moves up to claim the system optimum. 
We prove that the pricing problem can be decomposed and fares optimised at a ride level are also optimal for the system.

In matching, from the set of all feasible graphs, the operator chooses an optimal subset with respect to a objective function $O$ such that: each traveller is offered exactly one ride (\ref{eq:match_1}); if a shared ride is offered, it is assigned to all its co-travellers (\ref{eq:match_2}); the objective function $O$ is maximised (\ref{eq:match_3}). 

Let $D$ denote the set of all travellers. 
We follow notation from Section \ref{sec:objective_measure}.
Matching is a function $M: D \xrightarrow{} G$ that satisfies 
\begin{subequations}
\label{eq:matching}
\begin{align}
    \sum_{r \in M(D)} |T(r)| = |D|, \label{eq:match_1} \\
    \forall_{r \in M(D)} t \in T(r) \implies \forall_{t_c \in T(r)} M(t_c) = r. \label{eq:match_2}
\end{align}
\end{subequations}
Intuitively, condition \ref{eq:match_1} requires that all travellers receive an offer and condition \ref{eq:match_2} that the offer is concise, i.e. if a shared ride is offered, it is to all its participants.

Let $O_r: (G, \mathbb{R}^{|T(r)|}) \ni (r, (\lambda_1, \ldots, \lambda_{|T(r)|})) \xrightarrow{} \mathbb{R}$ be an objective function for $r$ (essentially any function which has maximum).
We define $O$ as $O(r, (\lambda_1, \ldots, \lambda_{|T(r)|})) = O_r(r, (\lambda_1, \ldots, \lambda_{|T(r)|}))$.
In our setting, the objective function $O$ is defined as in Section \ref{sec:offer}.

For a given pricing $\Lambda: G \ni r \xrightarrow{} (\lam{1}, \ldots, \lam{|T(r)|})$, each ride $r \in G$ has its evaluation with the objective function, i.e. $O(r, (\lambda_1, \ldots, \lambda_{|T(r)|}))$.
The optimal matching $M^*$ is such matching that if $M_0$ is a matching,
\begin{equation}\label{eq:match_3}
    \sum_{r \in M^*(D)} O(r, \Lambda(r)) \geq \sum_{r \in M_0(D)} O(r, \Lambda(r)).
\end{equation}
We say that the pricing $\Lambda^*$ is locally optimal, if for any other pricing $\Lambda^0$,
\begin{equation*}\label{eq:pricing_optimal}
    \forall_{r \in G} \Lambda^* (r) \geq \Lambda^0(r).
\end{equation*}

Once the optimal discounts are determined, we obtain a unique evaluation (pricing) for each ride $r \in G$.
A pricing $\Lambda^*$ is globally optimal if
\begin{equation}
    \forall_{\Lambda \mathrm{ - pricing}} \forall_{M^0 \mathrm{ - matching}} \exists_{M^* \mathrm{ - matching}} \sum_{r \in M^*(D)} O(r, \Lambda^*(r)) \geq \sum_{r \in M^0(D)} O(r, \Lambda^0(r))
\end{equation}

\begin{theorem}
    Locally optimal pricing is globally optimal.
\end{theorem}

\begin{proof}
    Let $\Lambda^*$ be a locally optimal pricing and $\Lambda^0$ a globally optimal. 
    By contradiction to the thesis, assume 
    \begin{equation*}
        \exists_{M^0 \mathrm{ - matching}} \forall_{M^* \mathrm{ - matching}} \sum_{r \in M^*(D)} O(r, \Lambda^*(r)) < \sum_{r \in M^0(D)} O(r, \Lambda^0(r)).
    \end{equation*}
    In particular,
    \begin{equation*}
        \sum_{r \in M^0(D)} O(r, \Lambda^*(r)) < \sum_{r \in M^0(D)} O(r, \Lambda^0(r)).
    \end{equation*}
    It implies that $\exists_{r \in M^0(D)} O(r, \Lambda^*(r)) < O(r, \Lambda^0(r))$ which contradicts local optimal of $\Lambda^*$. 
\end{proof}

We conclude that optimisation at a ride level leads to a system optimum. 
Our pricing maximises the expected profitability at a ride level (local optimum). 
Hence, the proposed method claims the maximum system performance.

\subsection{Illustrative example}
The proposed mechanism, based on the expected profitability, incorporates a variety of cases and sub-measures.
Here, we show evaluation of a single shared ride under two pricing methods.
We start with computations for a conventional flat-discount pricing and we follow by our personalised strategy.

Travellers $A$ and $B$ requested a pooled ride. 
They are well aligned and viable for the system, the operator offers a guaranteed discount of $5\%$. 
Their requested trips are for $3.6$ and $3.2$km while the shared ride takes $4.8$km. 
If both were offered a flat $20\%$ sharing discount, they would accept the ride with $70\%$ and $95\%$ probability.
The revenue resulting from the successful ride (if accepted by both customers) is $(3.6+3.2)*1.5*(1 - 0.2) = 8.16$ at a $0.7*0.95 = 0.665$ probability.
We need to account for reject decisions which are committed with $30$ and $5\%$ probability, respectively.
A scenario in which traveller $A$ rejects and $B$ accepts delivers a revenue of $3.6*1.5 + 3.2*1.5*(1-0.05) = 9.96$ and its probability is $0.3 * 0.95 = 0.285$. 
Analogically, a scenario in which $A$ accepts and $B$ rejects delivers a revenue of $3.6*1.5*(1-0.05) + 3.2*1.5 = 9.93$ at a $0.7*0.05 = 0.035$ probability. 
Finally, when both reject ($0.3*0.05=0.015$ chance) is associated with the revenue of $3.6*1.5 + 3.2*1.5 = 10.2$. 
Therefore, the expected revenue is $8.16*0.665 + 9.96*0.285 + 9.93*0.035 + 10.2*0.015 = 8.76555$. 
The expected profitability is a ratio of the expected revenue to the expected mileage (same scenarios) and equals $\frac{8.76555}{4.8*0.665 + (3.6+3.2)*(1-0.665)} = \frac{8.76555}{5.47} \approx 1.6024$ in this case.

In our personalised scheme, the traveller $A$, who experiences a greater inconvenience, receives a higher discount.
Travellers are offered $21.5\%$ and $13.8\%$, respectively. 
Their acceptance probability product raises to $0.72$ (traveller $A$ accepts with $80\%$ and $B$ with $90\%$ probability).
The calculated expected revenue is $8.406$ and the expected mileage equals $5.39$ resulting in the expected profitability improved to $1.642$.
The ride was purposefully chosen as a non-extreme example from our simulation. 
It was selected in the matching in both flat $20\%$ and individualised pricing schemes. 
Notably, the personalised pricing improves the expected profitability via increasing the matching probability for well-aligned trips.

\subsection{Summary}
We can summarise our approach as follows\footnote{Our method is publicly available and reproducible at available at \url{https://github.com/michallbujak/RidePoolingFrontiers}.}. 
We collect trip requests and obtain the distribution of behavioural characteristics in the population. 
With the information gathered, we apply the ExMAS algorithm to construct a dense shareability graph based on behavioural quantiles. 
For each traveller-ride configuration, we construct acceptance probability function of the individual discounts. 
To make calculations feasible, we discretize the search space. 
To find the optimal fares, we use a product of the acceptance probability with associated revenues (expected profitability). 
As a result, for each ride in the shareability graph, we obtain a vector of individual discounts along with the expected profitability of the ride. 
Next, we apply an existing (standard) solution to the linear integer programming problem to construct an offer.

\subsection{Alternative formulation of objective function}\label{sec:generalised_formula}
An operator can tailor their objective measure to target a specific market conditions. 
Now, to generalise, we present a formulation that incorporates four components: profitability, revenue, mileage and vehicle-ready cost. 
We introduce the generalised objective measure (for non-shared and shared rides, respectively) as follows.
\begin{subequations}\label{eq:alternative_profit}
    \begin{align}
    &\Gamma_{ns, i}(\hat{\lambda}) = \alpha_{ns, 0} \frac{\gamma_{ns,i}(\gd)}{d_i} + \alpha_{ns, 1} \gamma_{ns,i}(\gd) - \alpha_{ns, 2} d_i - C_{ns}  \\
    &\Gamma_{s, r_l}(\lambdas) = \alpha_{s, 0} \frac{\gamma_{s, r_l}(\lambdas)}{\psi_{s, r_l}(\lambdas)} + \alpha_{s, 1} \gamma_{s, r_l}(\lambdas) \\ 
    & - \alpha_{s, 2} \psi_{s, r_l}(\lambdas) - C_s,
\end{align}
\end{subequations}
where $\gamma_{ns, i}$ ($\gamma_{s, r_l}$) express the expected revenue, $d_i$ ($\psi_{s, r_l}(\lambdas)$) denotes the expected mileage ($d_i$ is deterministic for a non-shared ride) and $C_{ns}$ ($C_s$) account for a flat cost for a vehicle ready for service.
The relative significance of the components is controlled with the profitability sensitivity $\alpha_{ns, 0}$ ($\alpha_{s, 0}$), the revenue sensitivity $\alpha_{ns, 1}$ ($\alpha_{s, 1}$), the cost sensitivity $\alpha_{ns, 2}$ ($\alpha_{s, 2}$) and the value of the flat cost $C_{ns}$ ($C_s$).
We note that driver's commission and fleet expenses can be directly incorporated into $\alpha_{ns, 1}$, $\alpha_{ns, 2}$ and $C_{ns}$. 
The objective formulation in eq. \ref{eq:alternative_profit} can substitute the expected profitability (eq. \ref{eq:pricing_optim}) in the optimisation and matching process.

While the general formulation accounts for the market specifics, it is parameter-sensitive.
To provide a parameter-free analysis, we focus on the expected profitability ($\alpha_{ns, 1}$, $\alpha_{ns, 2}$, $C_{ns}$ and the sharing sensitivity equivalents set to $0$).

\section{Numerical study}\label{sec:numerical_study}
To present our method, we adopt a realistic setting. 
We set our experiment in NYC and use historic trip requests combined with preference studies reported in literature.
In the probabilistic setting, we benchmark our individualised pricing against flat discount strategies ($15\%$ and $20\%$). 
We look at four aspects. 
First, in the individual pricing, we investigate how discounts vary across different ride options.
Then, we scrutinise offers with respect to the private and shared rides distribution.
Third, we compare the economic performance. 
Finally, we analyse a traveller's perspective of the service.

\subsection{Experiment design}
We conduct our experiment in Manhattan, New York City, based on publicly available data published by the NYC Taxi and Limousine Commission \citep{tlcnycrecords}. 
To pick a representative sample, we rely on our previous results \citep{shulika2024spatiotemporal} and choose a medium density (300 requests/hour) 30-min batch.

For behavioural preferences of travellers, we rely on results from \cite{alonso2021determinants}. 
While the study was conducted elsewhere (Netherlands) and the utility formulation was slightly different, the data is accommodated to our model. 
Whenever a new, more adequate data is available, one can fine-tune the parameters in the experiment.
According to the preference study, the population is divided into four classes: LC1, LC2, LC3 and LC4. 
Each class exhibits different behavioural traits: the value of time and the penalty for sharing. 
We assume that properties within each class are normally distributed. 
In our recent study \citep{bujak2024ride}, we accommodated the data to our model, hence here we resort to presenting the resulting classes' properties in Table \ref{tab:vot_pfs_values}.
\begin{table}[!ht]
    \centering
    \resizebox{0.7\textwidth}{!}{
    \begin{tabular}{cc|c||c|c|c|c}
          & &  & LC1 & LC2 & LC3 & LC4 \\ \hline         
        \multirow{2}{*}{Value of time} & \multirow{2}{*}{$\beta_t$} & Mean & 16.98 & 14.02 & 26.25 & 7.78 \\ 
        & & Std.dev. & 0.318 & 0.201 & 5.777 & 1 \\ \hline
        \multirow{2}{*}{Penalty for sharing} & \multirow{2}{*}{$\beta_s$} & Mean  & 1.22 & 1.135 & 1.049 & 1.18 \\
        & & Std.dev. &  0.082 & 0.071 & 0.06 & 0.076 \\ \hline
        Share & $p_C$ & & 29\% & 28\% & 24\% & 19\% \\ \hline
    \end{tabular}}
    \caption{\small Behavioural parameters used in the analysis. Each client belongs with a given probability $p_C$ to a class LC1, LC2, LC3 or LC4. Each class is characterised by a two-dimensional normal distribution (with covariance 0) of value of time and penalty for sharing with parameters indicated in the table.}
    \label{tab:vot_pfs_values}
\end{table}

As for other arguments, we set the guaranteed discount $\hat{\lambda}$ to $5\%$ (as in \cite{uberpool}).
To construct a dense yet not overly optimistic shareability graph (according to equations \ref{eq:utilities_exmas} and \ref{eq:quantile_flat}), we set $\alpha=0.2$ and $\lambda = 0.4$. 

For a reasonable discount baseline, we calculate an average discount for the offer resulting from our algorithm for shared rides.
We obtain an average of $15\%$.
To broaden the analysis, we introduce a second flat strategy with $20\%$ discount.

\subsection{Distribution of personalised discounts} \label{sec:distribution_discounts}
Our pricing algorithm calculates individualised discounts at a ride level. 
In this section, we look at how diverse the personalised discounts are and how do they differ from a point-distribution of a flat discount. 
We distinguish two kinds of statistics. 
First, we consider all feasible rides in the shareability graph.
Second, we limit the analysis to the offer only.

\begin{figure}[!ht]
    \centering
    \includegraphics[width=0.8\linewidth]{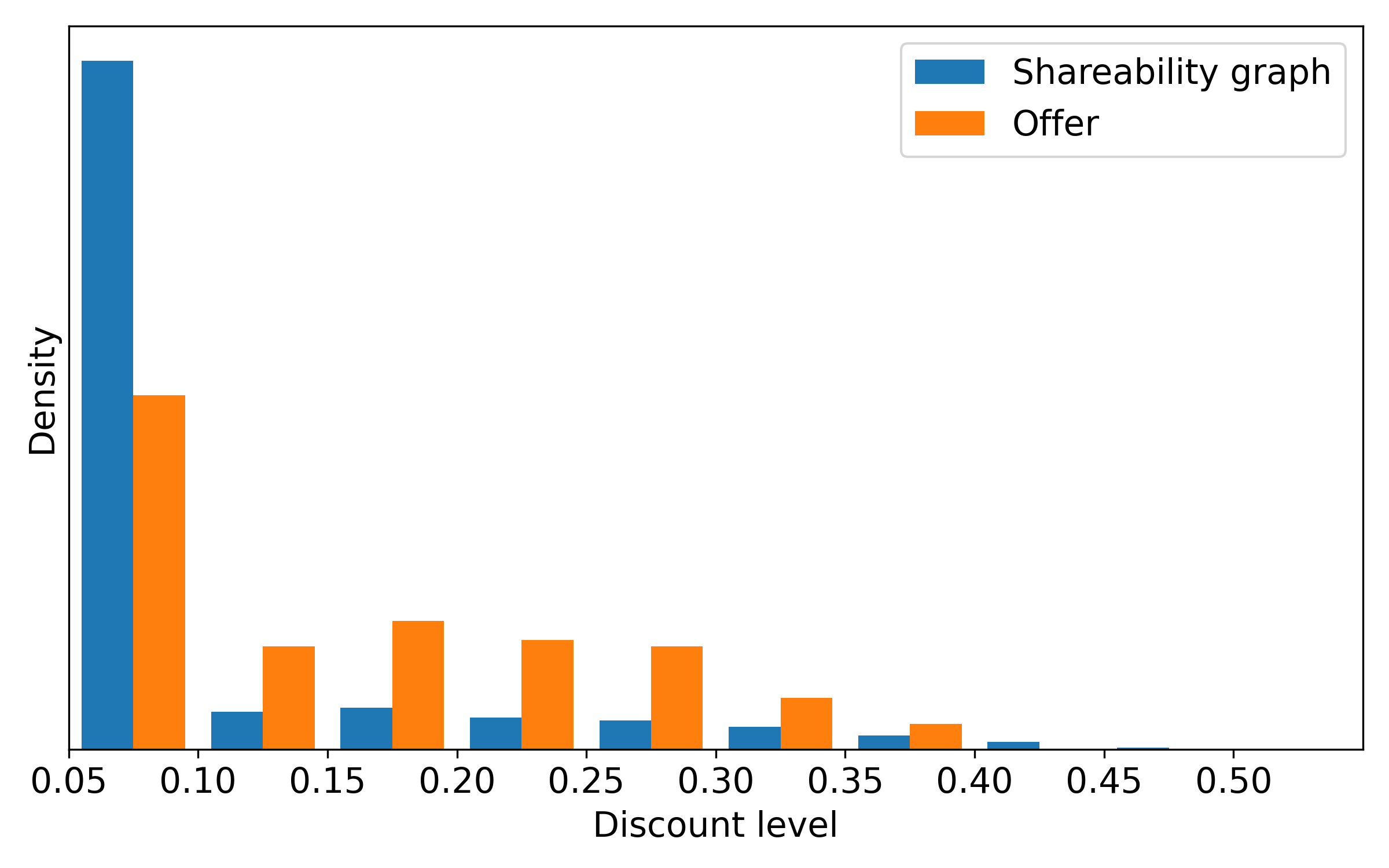}
    \caption{Density distribution of discounts resulting from the personalised pricing algorithm. We separately aggregate for all rides in the shareability graph (blue) and those selected for the offer (orange). Each bin represents a $5\%$ interval. While most rides in the shareability graph receive only the guaranteed minimum ($5\%$), those selected for the offer are associated with higher discounts. }
    \label{fig:discount_distribution}
\end{figure}
Majority of rides in the shareability graph reach maximum expected profitability when offered minimal discounts, as depicted in Figure \ref{fig:discount_distribution}.
Travellers are likely to reject shared rides and hence, the expected revenue is high. 
Despite the fact that higher discounts reduce the mileage (by increasing the acceptance probability), they rarely compensate for the revenue loss.
In the offer, rides with higher discounts are selected.
Efficient combinations where sharing yields significant mileage savings are worth a small sacrifice in the revenue.

\subsection{Offer structure}\label{sec:results_offer}
At the matching stage, in each pricing mechanism, we select the objective-maximising combinations from the shareability graph. 
Our objective is the expected profitability and for each ride depends on the vector of discounts. 
Here, we analyse characteristics of offers based on personalised and flat discount strategies.  

\begin{figure}[!ht]
    \centering
    \includegraphics[width=0.8\textwidth]{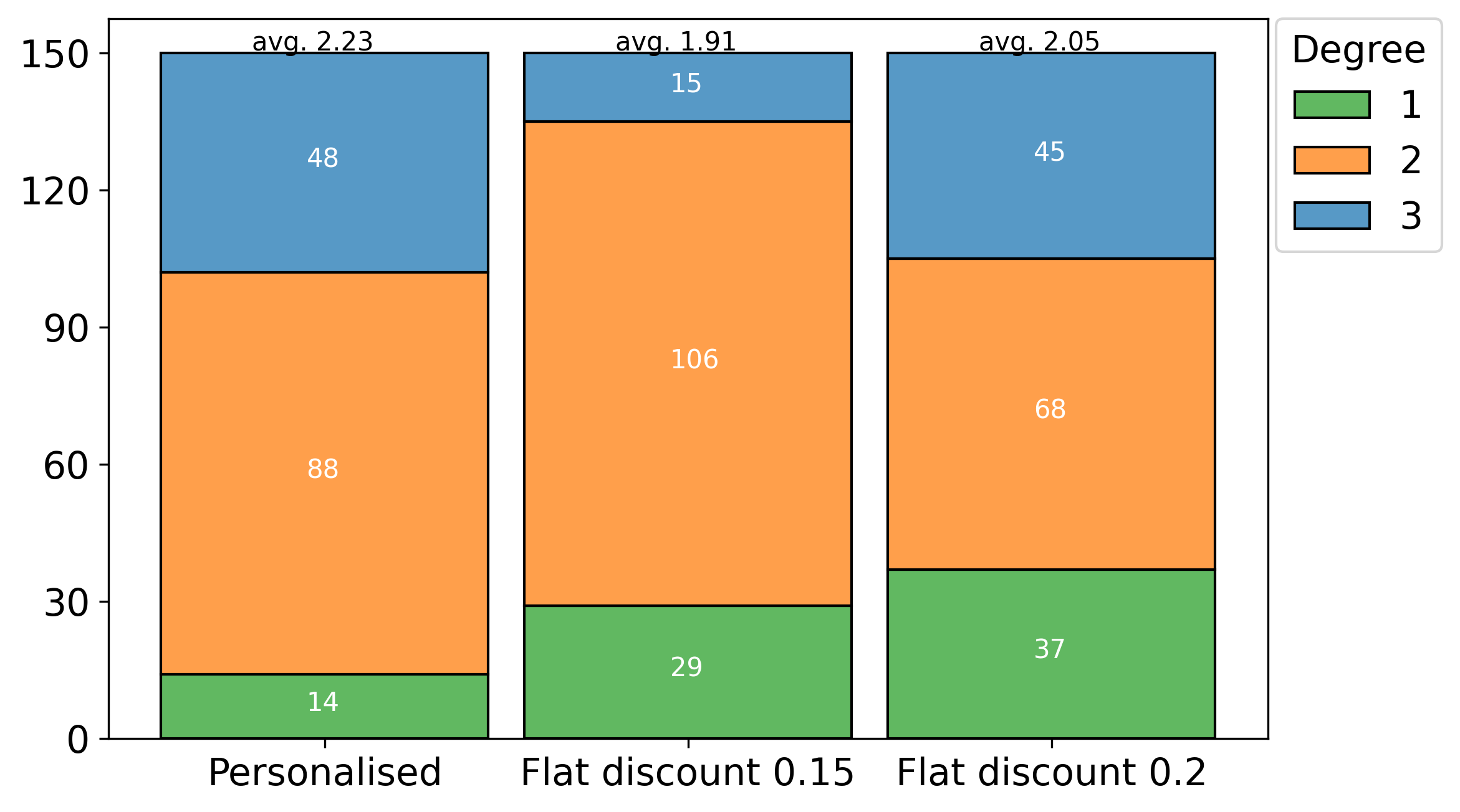}
    \caption{Degree distribution in offers constructed with 3 pricing strategies: personalised, flat sharing discount of $0.15$ and of $0.2$. The personalised approach favours more complex rides: only $14$ of $150$ travellers are offered a non-shared ride.}
    \label{fig:degrees}
\end{figure}

In Figure \ref{fig:degrees}, we can see degree distributions in the three pricing scenarios.
The number of shared rides is significantly increased with the personalised pricing.
This demonstrates that the decomposition works and the solution is globally effective.
The count of clients offered a private ride is only $14$ (out of $150$) compared to $28$ and $33$ under the flat discount strategies. 
The average degree of a ride is the highest for the personalised strategy ($2.23$) followed by $20\%$ discount ($2.05$) and $15\%$ ($1.91$).
None of the three strategies selected a ride of degree $4$ to the offer.

\subsection{Economy and system}\label{sec:economy}
When analysing performance of a ride-pooling system, two vital aspects are the economic efficiency and impact on the system. 
We measure the economic performance via the expected profitability.
To gain insights into system impact, we look at the expected mileage. 

We start by investigating the individual expected profitability for each shared ride in the shareability graph. 
We compare the three pricing strategies in Figure \ref{fig:profitability_unbalanced}. 
The personalised strategy is, by its definition, always (at least weakly) better. For each shared ride, the expected profitability is improved compared to flat baselines and above the private ride score ($1.425$ after a $5\%$ guaranteed discount). 
Rides of degree two often fall below the $1.425$ threshold under the flat pricing mechanisms. 
For more complex rides (degree $3$ and $4$), the acceptance probability is low, hence the expected profitability is driven primarily by undiscounted private rides (after a reject decision) and oscillates slightly above $1.425$.

\begin{figure}[!ht]
    \centering
    \begin{subfigure}{0.35\textwidth}
        \includegraphics[width=\textwidth, height=0.35\textheight]{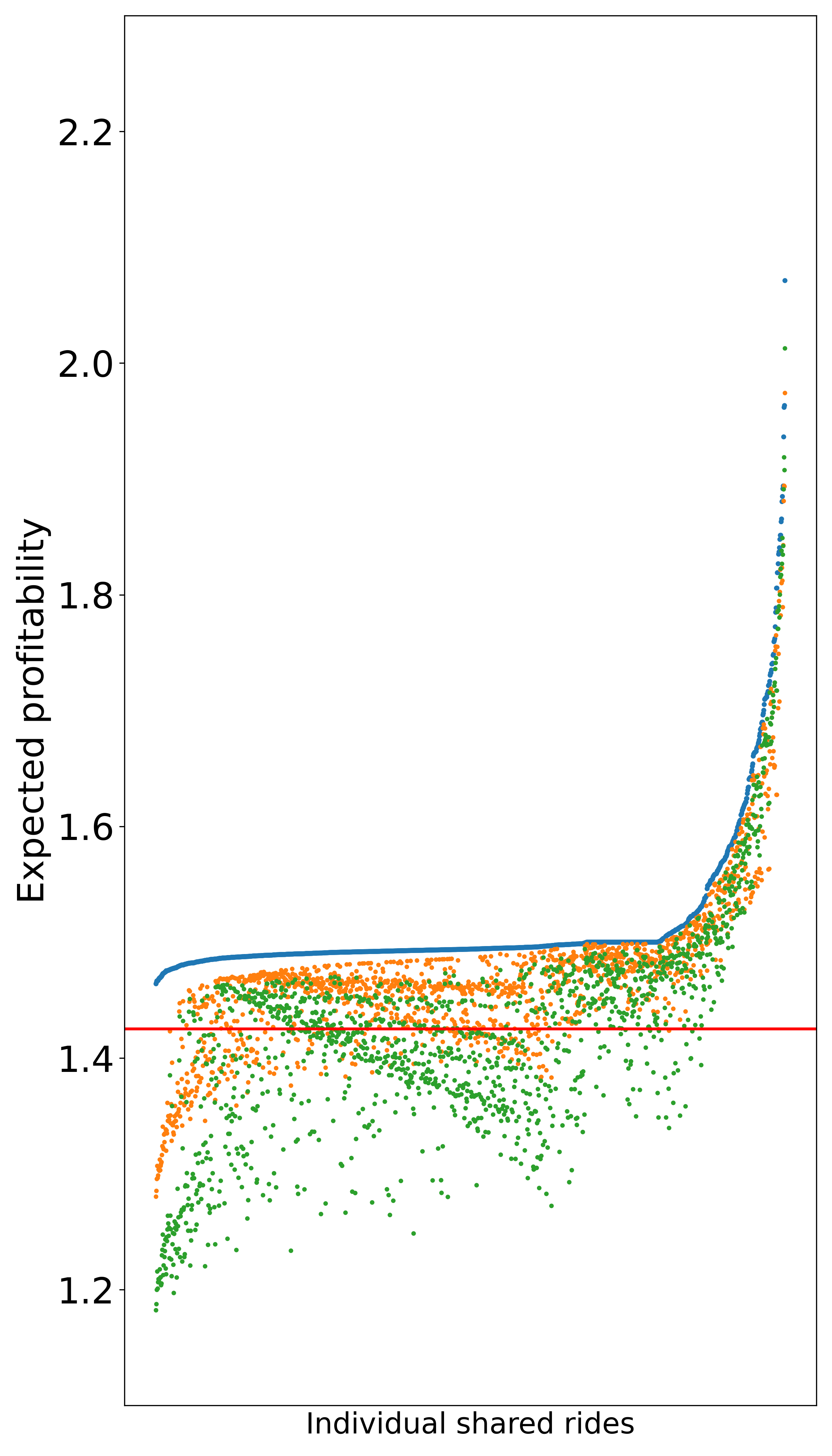}
        \caption{Degree 2}
        \label{fig:profitability_2}
    \end{subfigure}
    \begin{subfigure}{0.3\textwidth}
        \includegraphics[width=\textwidth, height=0.35\textheight]{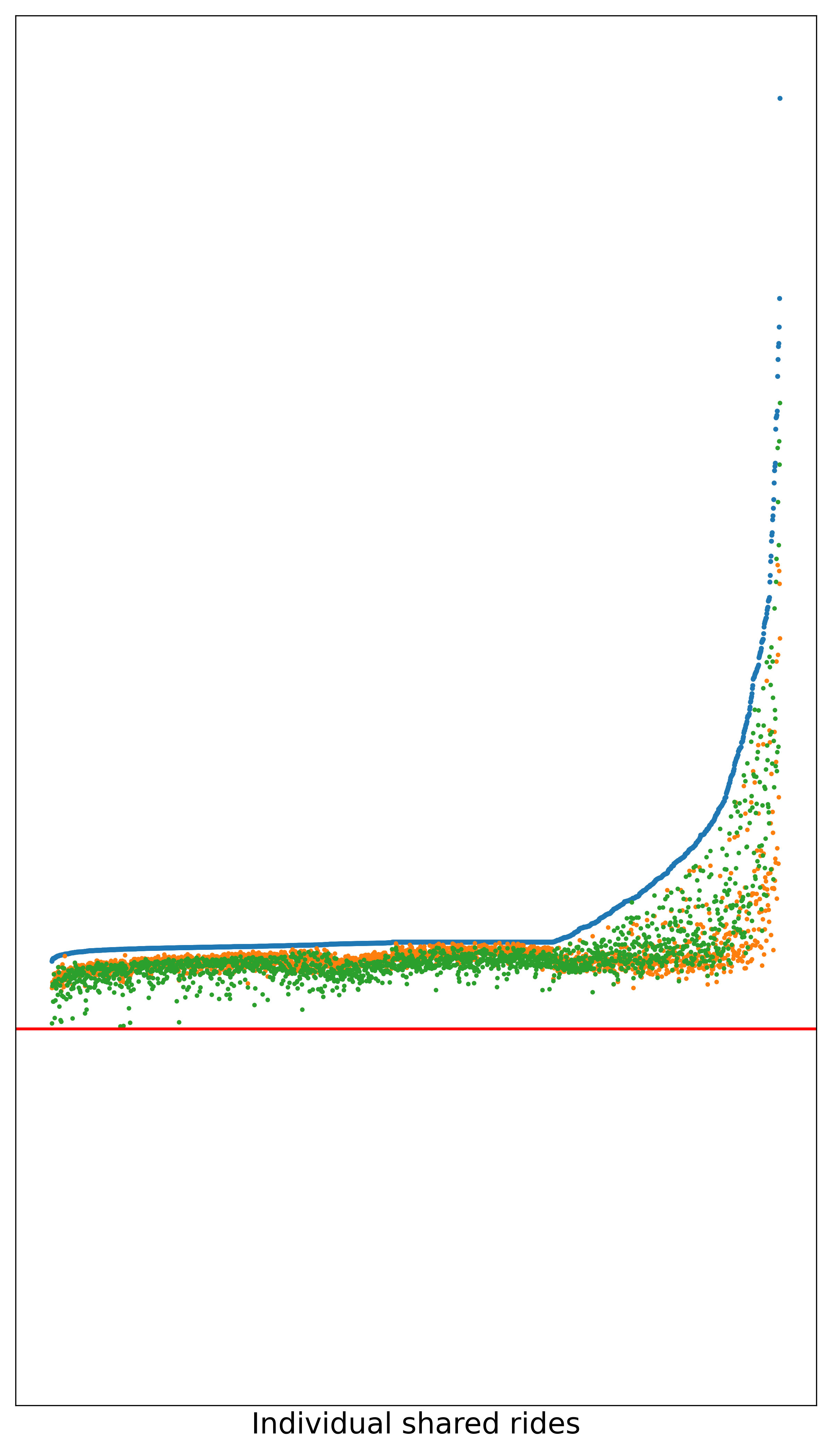}
        \caption{Degree 3}
        \label{fig:profitability_3}
    \end{subfigure}
    \begin{subfigure}{0.3\textwidth}
        \includegraphics[width=\textwidth, height=0.35\textheight]{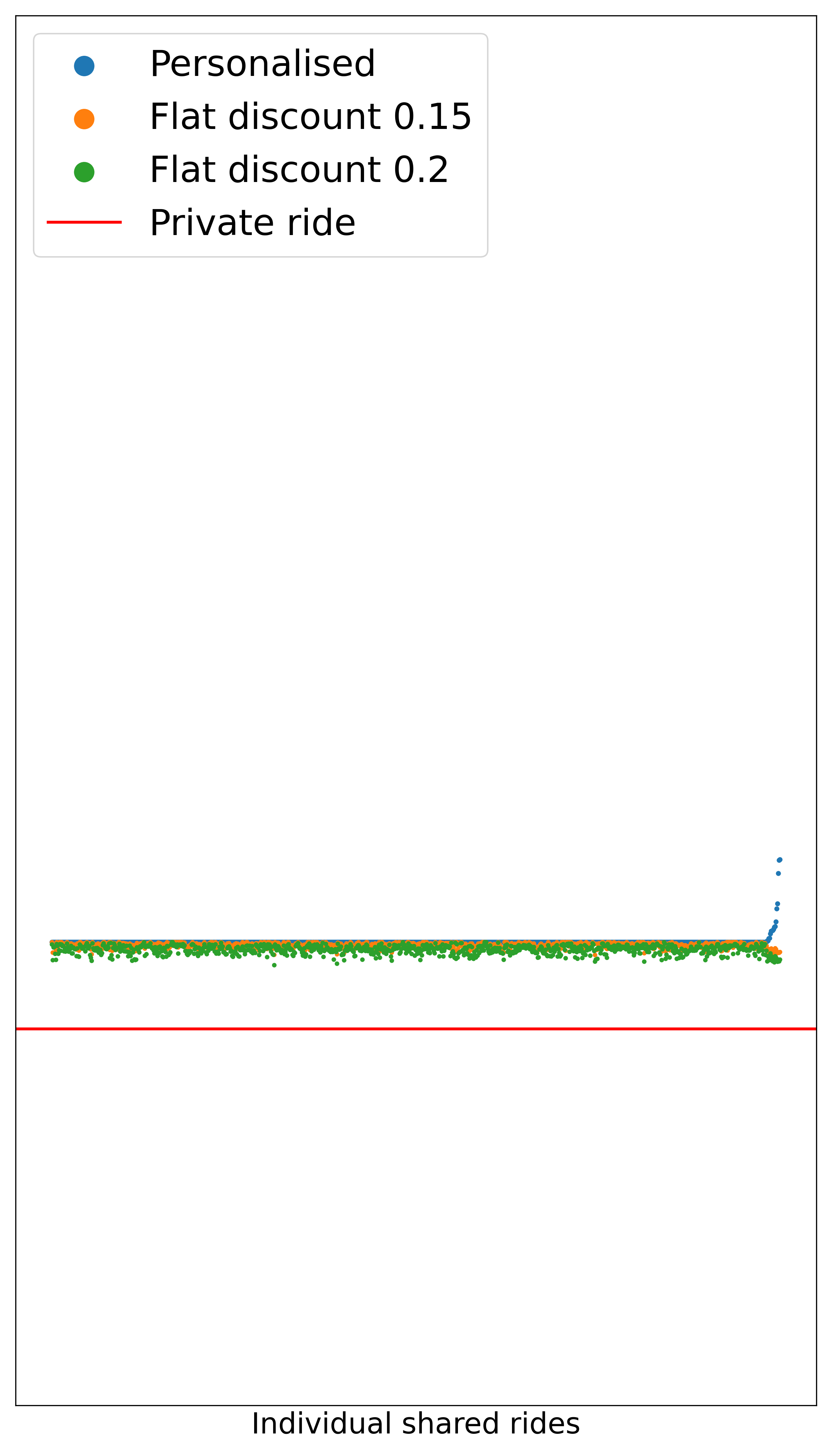}
        \caption{Degree 4}
        \label{fig:profitability_4}
    \end{subfigure}
    \caption{The expected profitability under pricing strategies: personalised and two flat discounts. At each vertical line, there are three points: blue, orange and green. They represent the expected profitability for the same shared ride under the three strategies, respectively. The red dotted line denotes $1.425$ threshold: the expected profitability of a discounted ($5\%$ guaranteed discount) private ride. 
    The personalised pricing leverages the local optimum, hence, for each feasible ride, the expected profitability is improved over the baselines. }
    \label{fig:profitability_unbalanced}
\end{figure}

We argue the personalised pricing strategy recognises both efficient and inefficient combinations. 
To investigate the ability, we analyse how relative mileage reduction correlates with expected profitability in Figure \ref{fig:profitability_distanced_saved}. 
The relative mileage reduction describes a percentage decrease in total distance when moving from private rides to a single shared rides.
Combinations that either increase mileage or yield small improvement do not exceed the $1.5$ threshold (profitability of an undiscounted private ride). 
We notice the first economic benefits associated with at least $10\%$ reduction. What follows resembles an exponential trend. 
Each step upwards the mileage reduction provides a substantial improvement in the expected profitability. 
We should note that around half of rides in the shareability graph offer at least $20\%$ mileage reduction.
\begin{figure}
    \centering
    \includegraphics[width=0.8\linewidth]{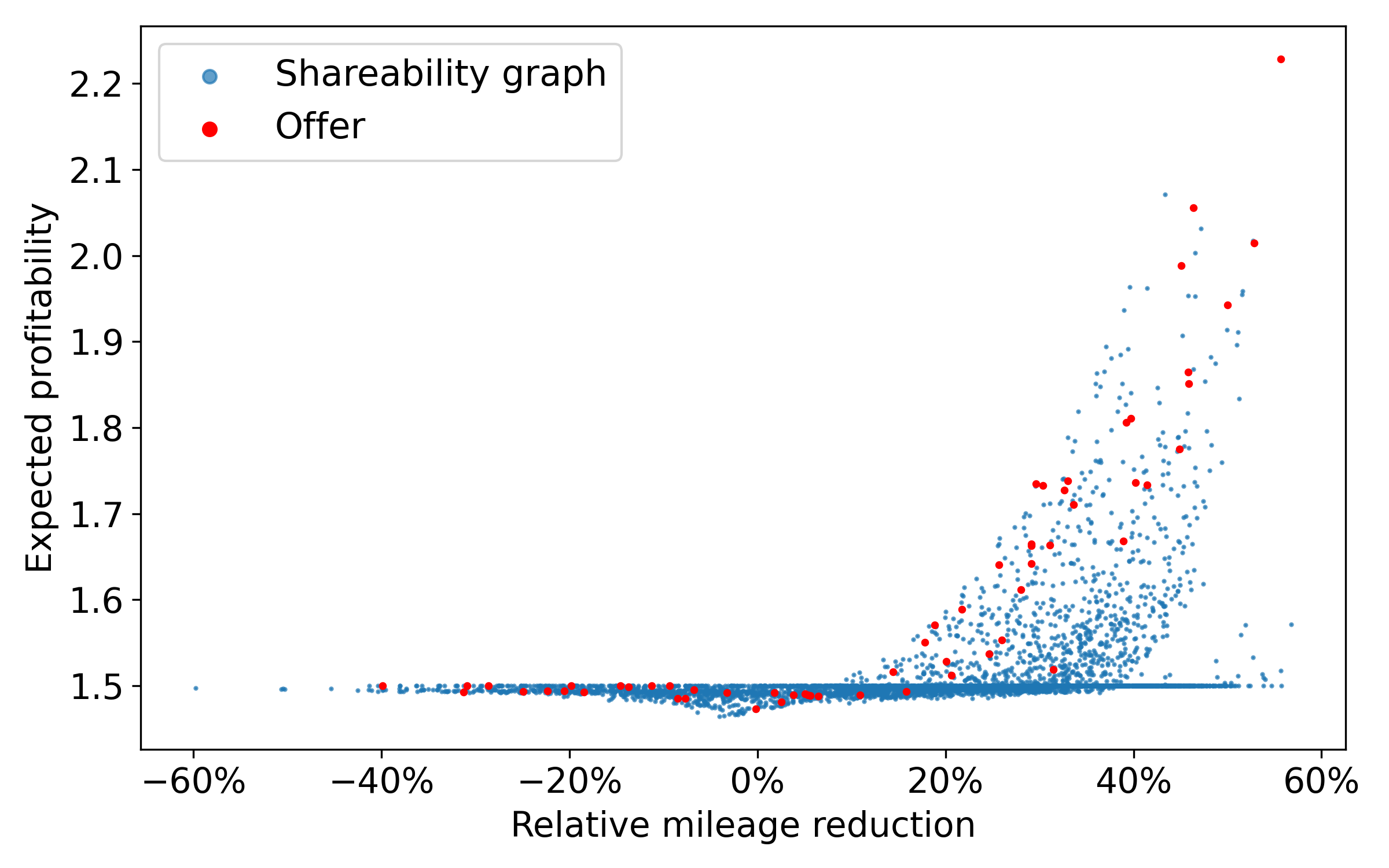}
    \caption{Distance reduction if a shared ride is realised (x-axis) against the expected profitability (y-axis). 
    Dots represent the individual ride's characteristics: in blue all rides in the shareability graph and in red those selected for the offer.
    To contribute to the system performance, rides need to be well aligned and offer mileage reduction.
    First efficient combinations (over $1.5$ profitability) appear around $10\%$ distance saved. 
    Upward the threshold, the trend resembles an exponential growth where each small increment in the mileage reduction results in a significant boost in the expected profitability.
    Around half of rides in the shareability graph provide at least $20\%$ distance reduction.
    }
    \label{fig:profitability_distanced_saved}
\end{figure}

We present aggregated economic and system results in Table \ref{tab:profit_distance}. 
In the first column, we present the average profitability of a ride in the system. In the second, we focus on the system's measure, the expected mileage. 
We additionally introduce a private ride-hailing baseline to investigate a general pooling potential.
Our measurements point to the personalised strategy as the best both in the economic and the system angle.

\begin{table}[!ht]
    \centering
    \begin{tabular}{lrr}
        \toprule
        {} &  Expected Profitability $\Gamma$ &  Expected Distance $\psi$ \\
        {} & (average) & (total) \\
        \midrule
        Personalised       &           3.29 &  332.20 \\
        Flat discount 0.15 &           2.70 &  358.11 \\
        Flat discount 0.2  &           2.73 &  348.06 \\
        Private only       &           1.50 &  415.39 \\
        \bottomrule
    \end{tabular}
    \caption{Comparison of three ride-pooling strategies and a private ride-hailing. The personalised strategy is superior both in economical and environmental perspective.}
    \label{tab:profit_distance}
\end{table}

\subsection{Service attractiveness}\label{sec:travellers}
We move to a next pillar of the ride-pooling service - customer satisfaction. 
The attractiveness perceived by travellers is a core factor that determines whether a shared ride is realised (everyone finds the ride attractive) or not (someone is dissatisfied with the ride).
In this section, we look at customer satisfaction under personalised and flat discount strategies. 
We investigate how do they correlate with the service performance measures.

In the probabilistic setting, we cannot use any deterministic measures. 
Thus, to quantify travellers' perspective, we use an acceptance probability.
The acceptance probability leverages population distribution to denote how likely a person is to find a certain shared ride attractive (given an appropriate fare according to a pricing).
We look separately at two stages. 
First, we analyse the attractiveness of potential shared rides (shareability graph) and second, in the offer (actually received propositions).
Results are presented in Figure \ref{fig:probs_accept}. 
We notice that most rides in the shareability graph are not attractive regardless of the pricing strategy. 
Operator, when conducting matching, selects two primary types of rides: only for pooling enthusiasts (low acceptance probability, low mileage reduction, high revenue) and efficient (high acceptance probability, high mileage reduction, lower revenue).
\begin{figure}[!ht]
    \centering    
     \begin{subfigure}[b]{0.50\textwidth}
         \centering
         \includegraphics[height=0.18\textheight, width=\textwidth]{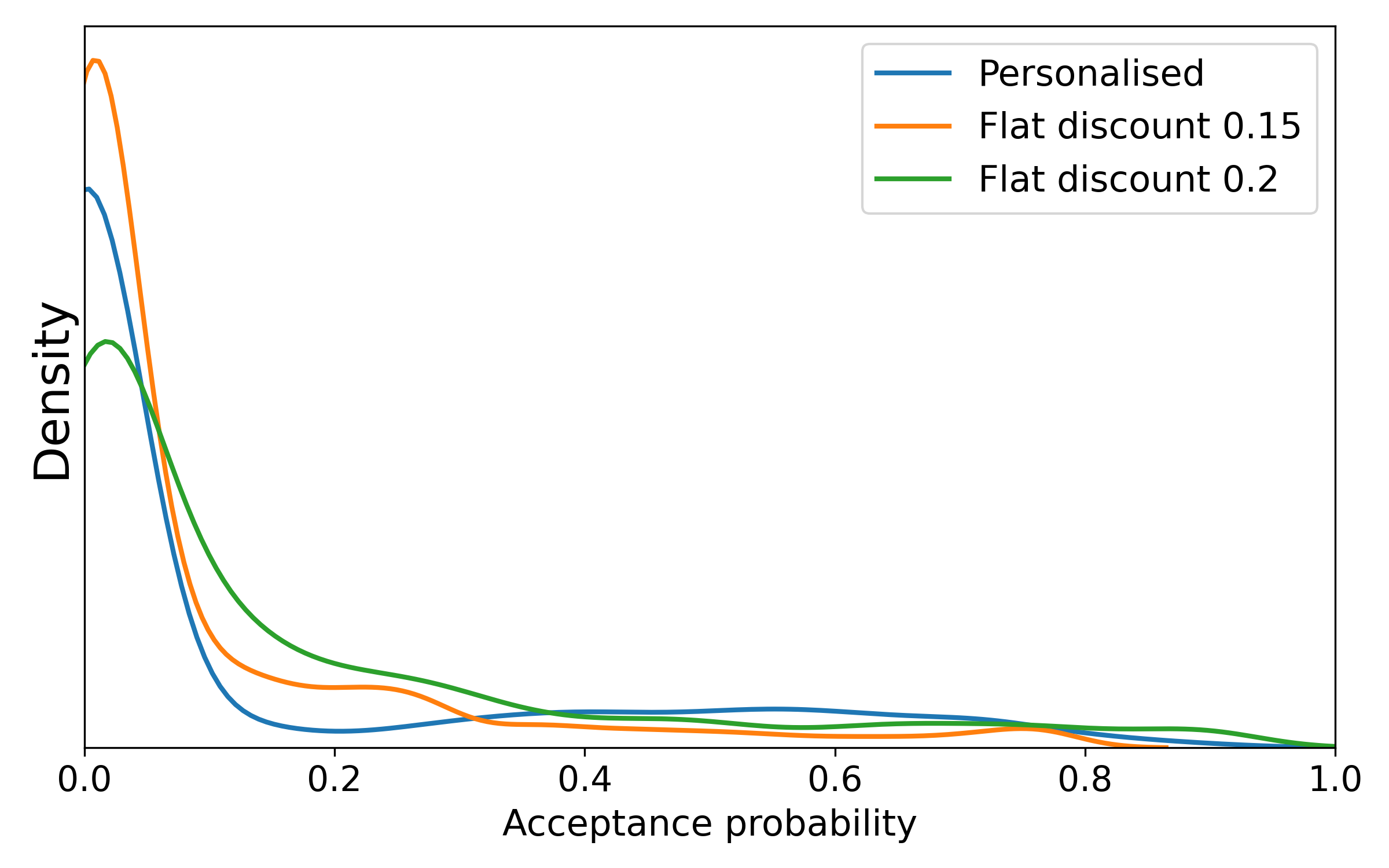}
         \caption{All rides in the shareability graph.}
         \label{fig:prob_all}
    \end{subfigure}
    \hfill
    \begin{subfigure}[b]{0.49\textwidth}
         \centering
         \includegraphics[height=0.18\textheight, width=\textwidth]{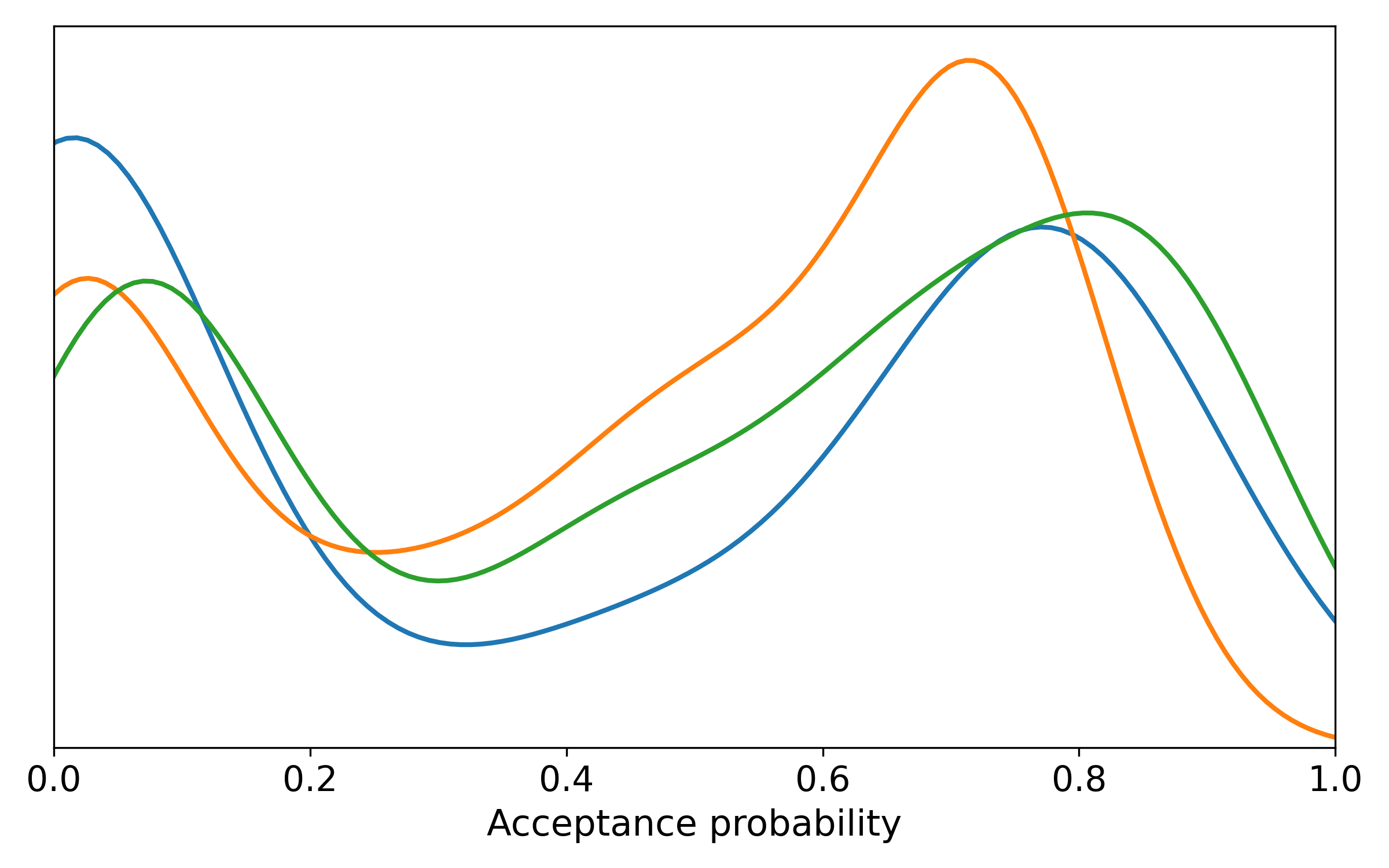}
         \caption{Rides selected in the offer.}
         \label{fig:prob_sel}
     \end{subfigure}
     \caption{Distribution of the acceptance probability for shared rides. 
     On the x-axis the acceptance probability and on the y-axis the distribution density. 
     Most rides in the shareability graph are not attractive and have a very low acceptance probability (Fig. \ref{fig:prob_all}). 
     For the offer, the operator selects two kinds (Fig. \ref{fig:prob_sel}) regardless of the pricing strategy.
     First, rides with a very low acceptance probability, which lead to undiscounted fares on private rides after rejection. 
     Second, shared rides likely to be realised.
     }
     \label{fig:probs_accept}
\end{figure}

We claim that personalised pricing is efficient in recognising well-aligned rides. 
To investigate the property, we investigate threefold relations between mileage reduction, profitability, and acceptance probability. 
Results depicted with heatmaps are presented in Figure \ref{fig:heatmaps}. 
We notice that the personalised pricing scheme efficiently assigns incentivising discounts to travellers who contribute to the system performance. 
Rides that offer over $30\%$ mileage reduction and high expected profitability (over $1.7$) are accepted with over $80\%$ probability.
Conversely, underperforming rides are offered a minimal discount.
We suspect that traveller’s role in the system is related to node’s properties in the shareability graph (examined in \cite{bujak2023network}). 
Flat discount strategies offer the same sharing discount regardless of a ride’s contribution to the system performance.

\begin{figure}[!ht]
    \centering
    \begin{subfigure}{0.32\textwidth}
        \includegraphics[width=\textwidth, height=0.2\textheight]{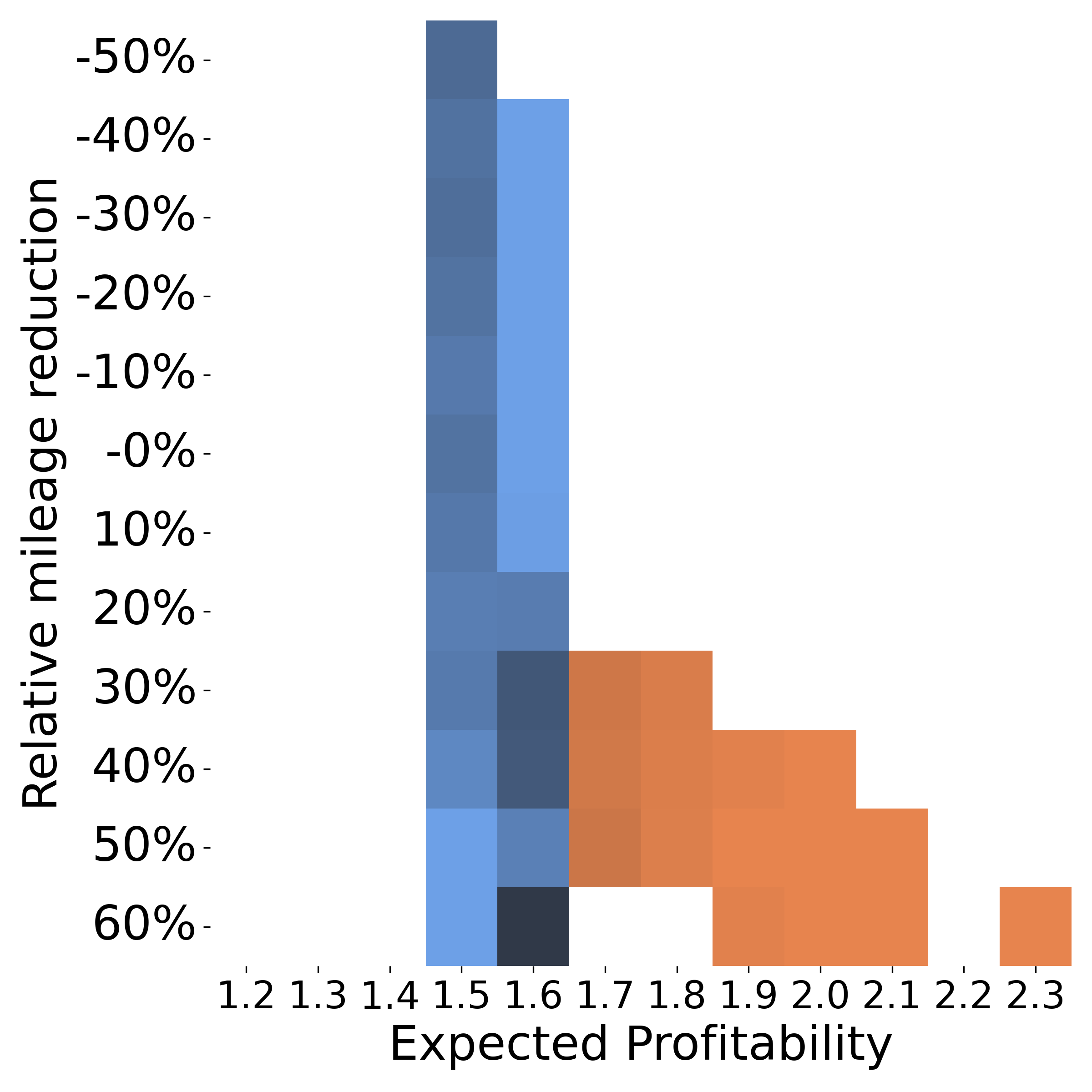}
        \caption{Personalised pricing}
        \label{fig:heatmap_pers}
    \end{subfigure}
    \begin{subfigure}{0.32\textwidth}
        \includegraphics[width=\textwidth, height=0.2\textheight]{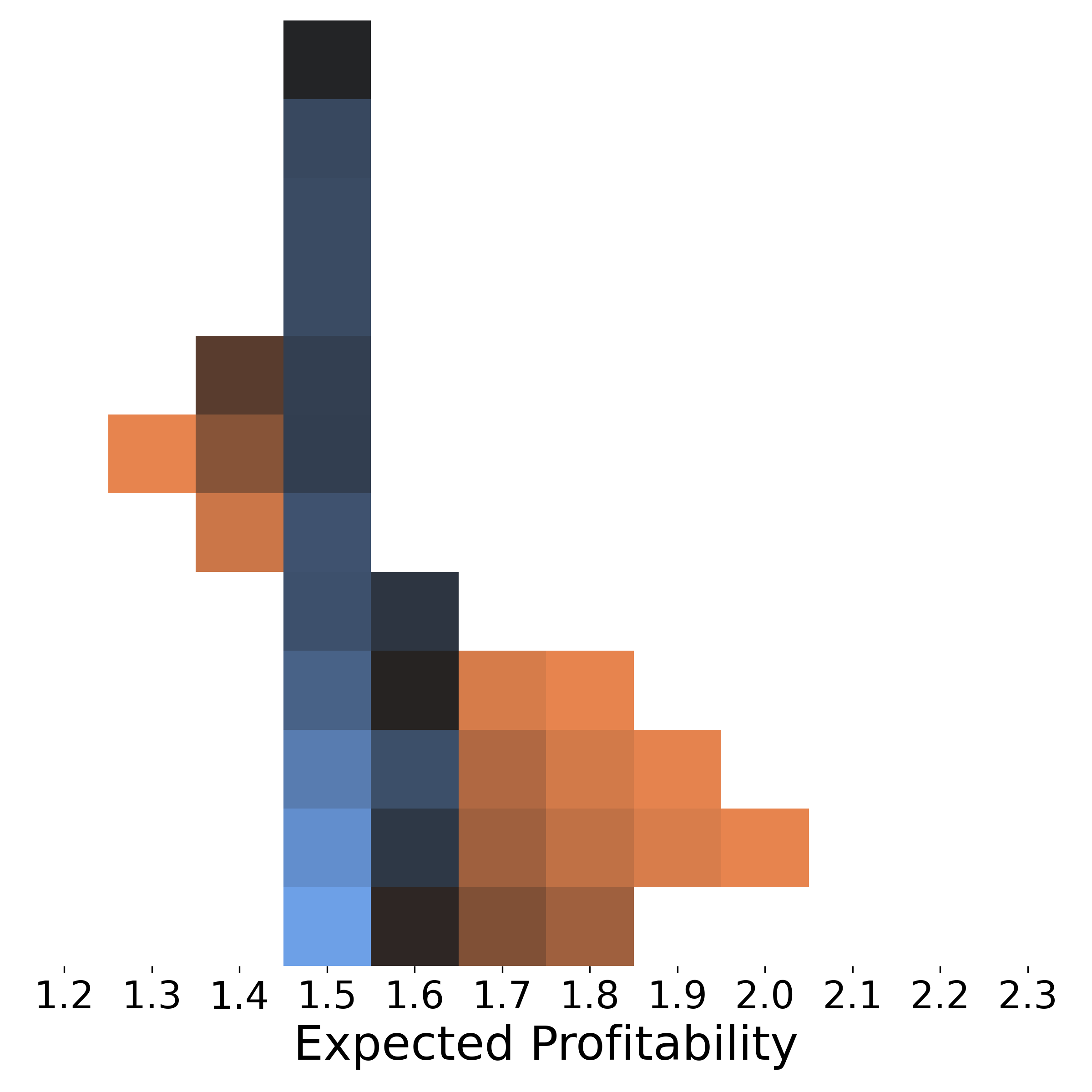}
        \caption{Flat $15\%$ discount}
        \label{fig:heatmap15}
    \end{subfigure}
    \begin{subfigure}{0.32\textwidth}
        \includegraphics[width=\textwidth, height=0.2\textheight]{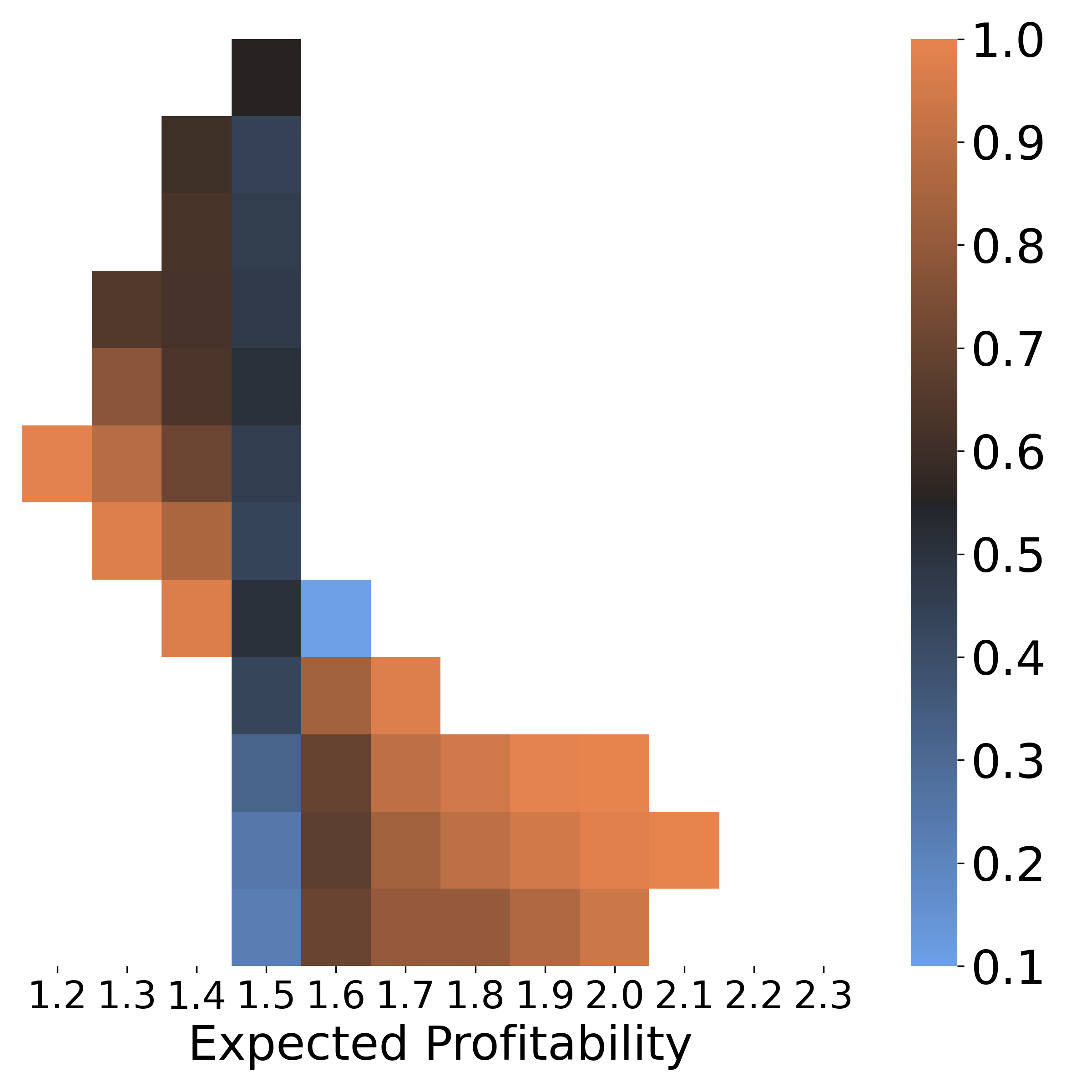}
        \caption{Flat $20\%$ discount}
        \label{fig:heatmap20}
    \end{subfigure}
    \caption{Relation between mileage reduction, expected profitability and acceptance probability under three pricing strategies for shared rides in the shareability graph.
    Personalised pricing assigns low discounts (low acceptance probability) to inefficient rides while incentivising well-aligned travellers with high discounts (high acceptance probability).}
    \label{fig:heatmaps}
\end{figure}

\section{Conclusions}
In this study, we propose a pricing strategy and address the hitherto neglected aspect in ride-pooling studies: population heterogeneity. 
We design a probabilistic framework with a two-stage optimisation process.
We introduce a new personalised pricing algorithm that maximises the ride-pooling service performance in the stochastic formulation.

In Section \ref{sec:methodology}, we propose the probabilistic framework.
We use a population distribution of behavioural traits to construct a shareability graph (Sec. \ref{sec:requests_shareability}). 
By scrutinising ride characteristic we develop an acceptance probability function for each traveller (Sec. \ref{sec:acceptance_probability}). 
To make the optimisation process feasible, we propose a discretization technique (Sec. \ref{sec:discretization}). 
In the ride-pooling system, travellers may not find a proposed shared ride attractive and rather choose private rides. 
In the evaluation process, we need to take both scenarios into account (Sec. \ref{sec:pricing_scheme}). 
To maximise the service performance at a ride level, we introduce the expected profitability measure that combines the per-mileage profit with the acceptance probability (Sec. \ref{sec:objective_measure}).
In Section \ref{sec:offer}, we explain how we choose feasible rides for the offer. 
In Section \ref{sec:loc_glob_optimum}, we prove that our optimisation approach is correct.
The pricing optimisation at a ride level leads to a system optimum. 
Finally, we describe a generalised objective formulation for the pricing (Sec. \ref{sec:generalised_formula}).

To evaluate our pricing mechanism in the proposed probabilistic setting, we conduct numerical simulation based on actual historic data (Sec. \ref{sec:numerical_study}). 
We find that the personalised strategy offers a wide range of discounts (Sec. \ref{sec:distribution_discounts}). 
In Figure \ref{fig:discount_distribution}, we show most feasible rides in the shareability graph do not qualify for preferential fares, i.e. they are poorly performing and receive the minimal discount. 
While analysing offers, we conclude that the personalised pricing favours more complex rides (Sec. \ref{sec:results_offer}, Fig. \ref{fig:degrees}) compared to baselines. 
In Section \ref{sec:economy}, we look at the service performance. 
We find that the personalised pricing strategy improves the expected profitability of each shared ride (Fig. \ref{fig:profitability_unbalanced}). 
In Figure \ref{fig:profitability_distanced_saved}, we show that our model restricts losses from inefficient combinations by maintaining close to private ride profitability while maximising economic gains from well-aligned trips. 
Combined with Figure \ref{fig:discount_distribution}, we understand that our model assigns minimal discounts whenever a shared ride does not contribute to the system efficiency. 
In Table \ref{tab:profit_distance}, we summarise the economic and system performance. The personalised approach is superior to flat discount strategies, improving the average expected profitability by $20\%$ and reducing the expected mileage by $4.5\%$. 
In Section \ref{sec:travellers}, we analyse travellers' perspective.
Figure \ref{fig:heatmaps} further highlight the personalised pricing property: inefficient rides are offered a minimal discount suitable only for pooling-enthusiasts while travellers in efficient rides receive a high incentive to join the service.

To summarise the contribution of our study. 
We develop a novel framework develop optimal pricing for a ride-pooling service. 
Within the framework, we propose a personalised pricing model that maximises economic and system benefits.
To make calculations feasible, we prove that the pricing process can be decomposed: optimisation conducted at a ride level leads to the global optimum.
The numerical results show that our method uncovers complex patterns in the ride-pooling: we identify rides and travellers crucial for the service performance. 
For a future research, we aim to incorporate behavioural feature estimation where the operator learns travellers' preferences with time.

\section*{Acknowledgements}
This research was co-funded by the European Union’s Horizon Europe Innovation Action under grant agreement No. 101103646.

This research is funded by National Science Centre in Poland program OPUS 19 (Grant Number 2020/37/B/HS4/01847).

\bibliography{references}

\begin{thebibliography}{31}
\providecommand{\natexlab}[1]{#1}
\providecommand{\url}[1]{\texttt{#1}}
\expandafter\ifx\csname urlstyle\endcsname\relax
  \providecommand{\doi}[1]{doi: #1}\else
  \providecommand{\doi}{doi: \begingroup \urlstyle{rm}\Url}\fi

\bibitem[Alonso-Gonz{\'a}lez et~al.(2020)Alonso-Gonz{\'a}lez, van Oort, Cats, Hoogendoorn-Lanser, and Hoogendoorn]{alonso2020value}
Mar{\'\i}a~J Alonso-Gonz{\'a}lez, Niels van Oort, Oded Cats, Sascha Hoogendoorn-Lanser, and Serge Hoogendoorn.
\newblock Value of time and reliability for urban pooled on-demand services.
\newblock \emph{Transportation Research Part C: Emerging Technologies}, 115:\penalty0 102621, 2020.

\bibitem[Alonso-Gonz{\'a}lez et~al.(2021)Alonso-Gonz{\'a}lez, Cats, van Oort, Hoogendoorn-Lanser, and Hoogendoorn]{alonso2021determinants}
Mar{\'\i}a~J Alonso-Gonz{\'a}lez, Oded Cats, Niels van Oort, Sascha Hoogendoorn-Lanser, and Serge Hoogendoorn.
\newblock What are the determinants of the willingness to share rides in pooled on-demand services?
\newblock \emph{Transportation}, 48\penalty0 (4):\penalty0 1733--1765, 2021.

\bibitem[Alonso-Mora et~al.(2017)Alonso-Mora, Samaranayake, Wallar, Frazzoli, and Rus]{alonso2017demand}
Javier Alonso-Mora, Samitha Samaranayake, Alex Wallar, Emilio Frazzoli, and Daniela Rus.
\newblock On-demand high-capacity ride-sharing via dynamic trip-vehicle assignment.
\newblock \emph{Proceedings of the National Academy of Sciences}, 114\penalty0 (3):\penalty0 462--467, 2017.

\bibitem[Bilali et~al.(2020)Bilali, Engelhardt, Dandl, Fastenrath, and Bogenberger]{bilali2020analytical}
Aledia Bilali, Roman Engelhardt, Florian Dandl, Ulrich Fastenrath, and Klaus Bogenberger.
\newblock Analytical and agent-based model to evaluate ride-pooling impact factors.
\newblock \emph{Transportation Research Record}, 2674\penalty0 (6):\penalty0 1--12, 2020.

\bibitem[Bujak and Kucharski(2023)]{bujak2023network}
Michal Bujak and Rafal Kucharski.
\newblock Network structures of urban ride-pooling problems and their properties.
\newblock \emph{Social Network Analysis and Mining}, 13\penalty0 (1):\penalty0 89, 2023.

\bibitem[Bujak and Kucharski(2024)]{bujak2024ride}
Michal Bujak and Rafal Kucharski.
\newblock Ride-pooling service assessment with heterogeneous travellers in non-deterministic setting.
\newblock \emph{Transportation}, pages 1--24, 2024.

\bibitem[Chavis and Gayah(2017)]{chavis2017development}
Celeste Chavis and Vikash~V Gayah.
\newblock Development of a mode choice model for general purpose flexible-route transit systems.
\newblock \emph{Transportation Research Record}, 2650\penalty0 (1):\penalty0 133--141, 2017.

\bibitem[Engelhardt et~al.(2019)Engelhardt, Dandl, Bilali, and Bogenberger]{engelhardt2019quantifying}
Roman Engelhardt, Florian Dandl, Aledia Bilali, and Klaus Bogenberger.
\newblock Quantifying the benefits of autonomous on-demand ride-pooling: A simulation study for munich, germany.
\newblock In \emph{2019 IEEE Intelligent Transportation Systems Conference (ITSC)}, pages 2992--2997, 2019.
\newblock \doi{10.1109/ITSC.2019.8916955}.

\bibitem[Fagnant and Kockelman(2014)]{fagnant2014travel}
Daniel~J. Fagnant and Kara~M. Kockelman.
\newblock The travel and environmental implications of shared autonomous vehicles, using agent-based model scenarios.
\newblock \emph{Transportation Research Part C: Emerging Technologies}, 40:\penalty0 1--13, 2014.
\newblock ISSN 0968-090X.
\newblock \doi{https://doi.org/10.1016/j.trc.2013.12.001}.
\newblock URL \url{https://www.sciencedirect.com/science/article/pii/S0968090X13002581}.

\bibitem[Ger{\v{z}}ini{\v{c}} et~al.(2023)Ger{\v{z}}ini{\v{c}}, van Oort, Hoogendoorn-Lanser, Cats, and Hoogendoorn]{gervzinivc2023potential}
Nejc Ger{\v{z}}ini{\v{c}}, Niels van Oort, Sascha Hoogendoorn-Lanser, Oded Cats, and Serge Hoogendoorn.
\newblock Potential of on-demand services for urban travel.
\newblock \emph{Transportation}, 50\penalty0 (4):\penalty0 1289--1321, 2023.

\bibitem[Ghasemi and Kucharski(2024)]{ghasemi2024modelling}
Farnoud Ghasemi and Rafal Kucharski.
\newblock Modelling the rise and fall of two-sided markets.
\newblock In \emph{Proceedings of the 23rd International Conference on Autonomous Agents and Multiagent Systems}, pages 679--687, 2024.

\bibitem[Jacob and Roet-Green(2021)]{jacob2021ride}
Jagan Jacob and Ricky Roet-Green.
\newblock Ride solo or pool: Designing price-service menus for a ride-sharing platform.
\newblock \emph{European Journal of Operational Research}, 295\penalty0 (3):\penalty0 1008--1024, 2021.

\bibitem[Karaenke et~al.(2023)Karaenke, Schiffer, and Waldherr]{karaenke2023benefits}
Paul Karaenke, Maximilian Schiffer, and Stefan Waldherr.
\newblock On the benefits of ex-post pricing for ride-pooling.
\newblock \emph{Transportation Research Part C: Emerging Technologies}, 155:\penalty0 104290, 2023.

\bibitem[Ke et~al.(2020)Ke, Yang, Li, Wang, and Ye]{ke2020pricing}
Jintao Ke, Hai Yang, Xinwei Li, Hai Wang, and Jieping Ye.
\newblock Pricing and equilibrium in on-demand ride-pooling markets.
\newblock \emph{Transportation Research Part B: Methodological}, 139:\penalty0 411--431, 2020.

\bibitem[Ke et~al.(2021)Ke, Zheng, Yang, and Ye]{ke2021probability}
Jintao Ke, Zhengfei Zheng, Hai Yang, and Jieping Ye.
\newblock Data-driven analysis on matching probability, routing distance and detour distance in ride-pooling services.
\newblock \emph{Transportation Research Part C: Emerging Technologies}, 124:\penalty0 102922, 2021.

\bibitem[Krueger et~al.(2016)Krueger, Rashidi, and Rose]{krueger2016preferences}
Rico Krueger, Taha~H. Rashidi, and John~M. Rose.
\newblock Preferences for shared autonomous vehicles.
\newblock \emph{Transportation Research Part C: Emerging Technologies}, 69:\penalty0 343--355, 2016.
\newblock ISSN 0968-090X.
\newblock \doi{https://doi.org/10.1016/j.trc.2016.06.015}.
\newblock URL \url{https://www.sciencedirect.com/science/article/pii/S0968090X16300870}.

\bibitem[Kucharski and Cats(2020)]{kucharski2020exmas}
Rafa{\l} Kucharski and Oded Cats.
\newblock Exact matching of attractive shared rides (exmas) for system-wide strategic evaluations.
\newblock \emph{Transportation Research Part B: Methodological}, 139:\penalty0 285--310, 2020.

\bibitem[Lavieri and Bhat(2019)]{lavieri2019modeling}
Patr{\'\i}cia~S Lavieri and Chandra~R Bhat.
\newblock Modeling individuals’ willingness to share trips with strangers in an autonomous vehicle future.
\newblock \emph{Transportation research part A: policy and practice}, 124:\penalty0 242--261, 2019.

\bibitem[Li et~al.(2022)Li, Jiang, and Lo]{li2022pricing}
Manzi Li, Gege Jiang, and Hong~K Lo.
\newblock Pricing strategy of ride-sourcing services under travel time variability.
\newblock \emph{Transportation Research Part E: Logistics and Transportation Review}, 159:\penalty0 102631, 2022.

\bibitem[Martinez and Viegas(2017)]{martinez2017assesing}
Luis~M. Martinez and José~Manuel Viegas.
\newblock Assessing the impacts of deploying a shared self-driving urban mobility system: An agent-based model applied to the city of lisbon, portugal.
\newblock \emph{International Journal of Transportation Science and Technology}, 6\penalty0 (1):\penalty0 13--27, 2017.
\newblock ISSN 2046-0430.
\newblock \doi{https://doi.org/10.1016/j.ijtst.2017.05.005}.
\newblock URL \url{https://www.sciencedirect.com/science/article/pii/S2046043016300442}.
\newblock Connected and Automated Vehicles: Effects on Traffic, Mobility and Urban Design.

\bibitem[Martinez et~al.(2015)Martinez, Correia, and Viegas]{martinez2015agent}
Luis~M. Martinez, Gonçalo H.~A. Correia, and José~M. Viegas.
\newblock An agent-based simulation model to assess the impacts of introducing a shared-taxi system: an application to lisbon (portugal).
\newblock \emph{Journal of Advanced Transportation}, 49\penalty0 (3):\penalty0 475--495, 2015.
\newblock \doi{https://doi.org/10.1002/atr.1283}.
\newblock URL \url{https://onlinelibrary.wiley.com/doi/abs/10.1002/atr.1283}.

\bibitem[Mitchell et~al.(2011)Mitchell, OSullivan, and Dunning]{mitchell2011pulp}
Stuart Mitchell, Michael OSullivan, and Iain Dunning.
\newblock Pulp: a linear programming toolkit for python.
\newblock \emph{The University of Auckland, Auckland, New Zealand}, 65, 2011.

\bibitem[Pandit et~al.(2019)Pandit, Mandar, Hanawal, and Moharir]{pandit2019pricing}
Vivek~Nagraj Pandit, Datar Mandar, Manjesh~K Hanawal, and Sharayu Moharir.
\newblock Pricing in ride sharing platforms: static vs dynamic strategies.
\newblock In \emph{2019 11th International Conference on Communication Systems \& Networks (COMSNETS)}, pages 208--215. IEEE, 2019.

\bibitem[Santi et~al.(2014)Santi, Resta, Szell, Sobolevsky, Strogatz, and Ratti]{santi2014quantifying}
Paolo Santi, Giovanni Resta, Michael Szell, Stanislav Sobolevsky, Steven~H Strogatz, and Carlo Ratti.
\newblock Quantifying the benefits of vehicle pooling with shareability networks.
\newblock \emph{Proceedings of the National Academy of Sciences}, 111\penalty0 (37):\penalty0 13290--13294, 2014.

\bibitem[Shah et~al.(2020)Shah, Lowalekar, and Varakantham]{shah2020neural}
Sanket Shah, Meghna Lowalekar, and Pradeep Varakantham.
\newblock Neural approximate dynamic programming for on-demand ride-pooling.
\newblock In \emph{Proceedings of the AAAI Conference on Artificial Intelligence}, volume~34, pages 507--515, 2020.

\bibitem[Shulika et~al.(2024)Shulika, Bujak, Ghasemi, and Kucharski]{shulika2024spatiotemporal}
Olha Shulika, Michal Bujak, Farnoud Ghasemi, and Rafal Kucharski.
\newblock Spatiotemporal variability of ride-pooling potential--half a year new york city experiment.
\newblock \emph{Journal of Transport Geography}, 114:\penalty0 103767, 2024.

\bibitem[Soza-Parra et~al.(2022)Soza-Parra, Kucharski, and Cats]{soza2022shareability}
Jaime Soza-Parra, Rafał Kucharski, and Oded Cats.
\newblock The shareability potential of ride-pooling under alternative spatial demand patterns.
\newblock \emph{Transportmetrica A: Transport Science}, 0\penalty0 (0):\penalty0 1--23, 2022.
\newblock \doi{10.1080/23249935.2022.2140022}.
\newblock URL \url{https://doi.org/10.1080/23249935.2022.2140022}.

\bibitem[Taxi and Commission(2024)]{tlcnycrecords}
NYC Taxi and Limousine Commission.
\newblock Tlc trip record data.
\newblock \url{https://www.nyc.gov/site/tlc/about/tlc-trip-record-data.page}, 2024.
\newblock Accessed: 2024-07-05.

\bibitem[UberXShare()]{uberpool}
UberXShare.
\newblock Uberx share.
\newblock \url{https://www.uber.com/us/en/ride/uberx-share/}, 2024.
\newblock Accessed: 2024-10-29.

\bibitem[Zhang and Nie(2021)]{zhang2021pool}
Kenan Zhang and Yu~Marco Nie.
\newblock To pool or not to pool: Equilibrium, pricing and regulation.
\newblock \emph{Transportation Research Part B: Methodological}, 151:\penalty0 59--90, 2021.

\bibitem[Zwick and Axhausen(2022)]{zwick2022ride}
Felix Zwick and Kay~W Axhausen.
\newblock Ride-pooling demand prediction: A spatiotemporal assessment in germany.
\newblock \emph{Journal of Transport Geography}, 100:\penalty0 103307, 2022.

\end{thebibliography}

\end{document}